\documentclass [a4paper, 11pt]{article}

\usepackage[margin=2.0cm]{geometry}
\usepackage{setspace}
\usepackage{amsmath}
\usepackage{amssymb}
\usepackage{amsthm}
\usepackage{dsfont}
\usepackage{gensymb}
\usepackage{bbm}
\usepackage{color}
\usepackage{multirow}
\usepackage[dvipsnames]{xcolor}
\usepackage{graphicx}
\usepackage{caption}
\usepackage{subcaption}
\usepackage{mathrsfs}
\usepackage{bbm}
\usepackage{xfrac}
\usepackage{pgfplots}
\usepackage{geometry}
\usetikzlibrary{intersections}
\usetikzlibrary{patterns}
\usetikzlibrary{arrows}
\usepackage{empheq}
\usepackage[skins,theorems]{tcolorbox}
\tcbset{highlight math style={enhanced,
  colframe=red,colback=white,arc=0pt,boxrule=1pt}}
  
\usepgfplotslibrary{fillbetween}
\usetikzlibrary{shapes,snakes,positioning}
\usetikzlibrary{automata,fit}
\usepackage{todonotes}
\usepackage{newtxmath}
\usepackage[T1]{fontenc}
\usepackage[bookmarksnumbered=true]{hyperref} 
\hypersetup{
     colorlinks = true,
     linkcolor = red,
     anchorcolor = blue,
     citecolor = blue,
     filecolor = blue,
     urlcolor = blue
     }
\usepackage{natbib}

\newtheorem{definition}{Definition}
\newtheorem{proposition} {Proposition}
\newtheorem{lemma} {Lemma}
\newtheorem{theorem} {Theorem}

\newtheorem{corollary}  {Corollary}

\title{Efficiency with(out) intermediation in\\ repeated bilateral trade}
\author{Rohit Lamba\footnote{Pennsylvania State University, rlamba@psu.edu. This paper is based on Chapter 1 of my PhD thesis at Princeton University, and a previous version was circulated under the title "Repeated Bargaining: A Mechanism Design Approach". I'm deeply indebted to my committee: Stephen Morris, Dilip Abreu, Marco Battaglini and Sylvain Chassang for their advice, comments and encouragement. I'm particularly thankful to Stephen for his guidance. Thanks also to Nemanja Antic, Anmol Bhandari, Ben Brooks, Bob Evans, Jacob Goldin, Ilia Krasikov, Jay Lu, David Martimort, Roger Myerson, Charles Roddie and Juuso Toikka for their comments. I also benefited from comments from seminar participants at New York University, Yale, Washington University St. Louis, Penn State University, University of Chicago, Indian School of Business, Indian Statistical Institute Delhi, Paris School of Economics, HEC Paris, H\"{u}mboldt University and University of Cambridge. Part of the work was done when I was a postdoctoral fellow at the Cambridge-INET institute at University Cambridge, and I'm grateful to colleagues for the wonderful research environment there. } \\ \scriptsize{First Version: November 2013}}
\date{February 2022}

\begin{document}

\maketitle

\begin{abstract}
This paper analyzes repeated version of the bilateral trade model where the independent payoff relevant private information of the buyer and the seller is correlated across time. Using this setup it makes the following five contributions. First, it derives necessary and sufficient conditions on the primitives of the model as to when efficiency can be attained under ex post budget balance and participation constraints. Second, in doing so, it introduces an intermediate notion of budget balance called interim budget balance that allows for the extension of liquidity but with participation constraints for the issuing authority interpreted here as an intermediary. Third, it pins down the class of all possible mechanisms that can implement the efficient allocation with and without an intermediary. Fourth, it provides a foundation for the role of an intermediary in a dynamic mechanism design model under informational constraints.  And, fifth, it argues for a careful interpretation of the "folk proposition" that less information is better for efficiency in dynamic mechanisms under ex post budget balance and observability of transfers.

\end{abstract}

\section{Introduction}

When can we achieve efficient trade under private information, voluntary participation and cap on external subsidies? In a seminal article \citet{myerson1983efficient} showed that under plausible assumptions on preferences and distribution of information, the answer to aforementioned question is \emph{never}.\footnote{One of the motivations of this idea was to generate a contradiction to Coase's Theorem. \citet{coase1960cost} argued that if transactions costs are low enough and trade a possibility, bargaining will eventually lead to an efficient outcome independent of the initial distribution of property rights.} This paper reprises the interaction of the same three institutional constraints (incentive compatibility, individual rationality and budget balance) elegantly modeled in \citet{myerson1983efficient}, in a dynamic context. It asks and precisely answers the question: when can efficiency be sustained in repeated bilateral trade?

Why study dynamic trade? Arguably, many economic transactions modeled in mechanism design are inherently dynamic, where information revealed today can be used to set contracts tomorrow. Public goods such as food subsidies are provided repeatedly. In a fast changing technological landscape, spectrum auctions and buybacks are taking place repeatedly. Taxation is often dynamic and tagged with age, social security being a case in point. Wage contracts and bonuses depend on performance parameters evaluated over time. Online selling can now rely on a huge treasure trove of past buying data. This paper contributes towards a theory of such dynamic institutions in what is now a recently burgeoned literature on dynamic mechanism design.

I study repeated version of the standard bilateral trade problem with two sided asymmetric information. The buyer has a hidden valuation for a non-durable good and the seller has a hidden cost of producing it, referred to as their types. These evolve over time according to independent Markov processes. Trade can happen with some probability and transfers are made to the agents in each period. Trade is voluntary in the sense that at the start of each period both agents expect a non-negative utility from continuing in the relationship. Finally, budget must be balanced.


At a technical level, an "open question" in the literature has been the identification of the set of parameters for which efficiency can be sustained with voluntary participation and ex post budget balance.\footnote{Here efficiency refers to the trading rule that whenever the valuation of the buyer is greater than the cost of production for the seller, trade should happen. Ex post budget balance simply refers to a rule that transfers to both agents add up to zero.} Intuitively speaking, for highly impatient agents, we must get inefficiency due to the aforementioned impossibility result in the static environment. For very high levels of "persistence" in types we must also get inefficiency because at constant types the optimal contract involves repetition of the static optimum. The first contribution of this paper is to completely characterize the set of parameters for which efficiency can be sustained in the repeated version of the bilateral trade problem.

While incentive compatibility and individual rationality have natural analogues in the dynamic setup, budget balance can be defined in a variety of ways, depending on the economic application at hand. Ex post budget balance may seem extreme in the presence of credit markets. An intermediary can offer to broker trade between the agents. Its job is to provide subsidy in bad times, and extract surplus in good times, making a profit in the bargain. Motivated from static mechanism design, a natural rule so far employed in the literature that involves a third party is ex ante budget balance. Each period both agents make payments to the intermediary or get paid by it. Ex ante budget balance demands that the initial expected net present value of these transfers is non-negative.

Ex ante budget balance demands a very strong form of commitment from the intermediary. The evolution of valuations (of the buyer) and production costs (of the seller) can lead the contract to a region where the intermediary expects to loose money by continuing to broker trade between the agents. In such a scenario shouldn't it be allowed to walk away or file for bankruptcy?

The second contribution of the paper is to introduce a new concept of budget balance which I term \emph{interim budget balance}. It allows for the role of an intermediary protected by participation constraints. After any history, the expected net present value of current and future cash flows from the buyer and the seller must be non-negative. Technically speaking, it imposes "ex ante budget balance" at every node of the contract tree. It is weaker than ex post budget balance, but stronger than ex ante budget balance. It can be regarded as a self-enforcing constraint that tempers the demands of commitment from the credit issuing authority.

I show that implementation under interim and ex post budget balance are \emph{equivalent} in the following sense. Fix any incentive compatible and individually rational trading rule. If it can be implemented under interim budget balance, then there exists a set of transfers that ensure that it is also implementable under ex post budget balance (converse of this is trivially true). The result is surprising because interim budget balance takes expectations over the entire future sequence of transfers, whereas ex post budget balance is bereft of dynamics---it is a static, distribution free concept.

In order to establish the equivalence result, I state and prove a \emph{history dependent version of the payoff equivalence result}. Consider two incentive compatible mechanism that implement the same allocation. The expected utility of the agents in both mechanisms must differ from each other through linearly additive history dependent numbers. Starting from an implementable interim budget balanced mechanism, payoff equivalence helps construct an implementable ex post budget balanced mechanism.

The third contribution is to provide an exhaustive set of mechanisms that can implement the efficient allocation under interim (and ex post) budget balance. Once the fact of efficiency has been established, there are multiple ways of distributing the efficient surplus across the buyer, seller and intermediary. The payoff equivalence result allows us to precisely capture all such bargaining protocols through a vector of history dependent weights.

I identity a particularly simple mechanism that implements efficient allocation under interim budget balance. Every period the buyer and the seller pay a Markov fee to the intermediary to avail of its services. The fees is Markov in the sense that it only depends on the agents' types from the last period. After accepting the fees, the intermediary runs the standard Vickery-Clarkes-Groves (VCG) mechanism. It is simple to understand since it employs a mixture of a fixed fee and VCG mechanism. It relatively detail free in that it does not depend on entire history of the contract, it has a one-period memory. I also show that with constant support of valuations and production costs the mechanism is also stationary.

While the equivalence result between interim and ex post budget balance is theoretically appealing, it can also lead the reader towards an interpretation that the intermediary is irrelevant. What, if anything, could explain the desirability of intermediation in dynamic trade? In the fourth contribution of the paper, I provide one such foundation---informational constraints.

What if the agents can access signals about the other agent's current type? What if information leakage or transmission is a concern for the design of the underlying institution? To accommodate these concerns the stronger notion of ex post incentive compatibility is invoked. I first show that for interim budget balance when the efficient allocation can be implemented under standard (Bayesian) incentive compatibility, it can also be implemented under (within period) ex post incentive compatibility. However, there is no such hope under ex post budget balance. Hence, with information sharing or inability to deter information snooping amongst the agents, intermediation can strictly improve upon welfare.

It is also plausible that information about the other's past types can be hidden from the agents, and only partial signals be made available to them, say in the form of prices. In such a situation too the equivalence between interim and ex post budget balance breaks down. I show that under observability of transfers (which is a reasonable assumption in many practical scenarios) this muted flow of information makes ex post budget balance highly restrictive. On the other hand, interim budget balance allows implementation for a larger class of parameters.\footnote{There could of course be other reasons for the desirability of intermediation such as credit constraints or lack of commitment. In this paper I focus only on informational constraints, parsing our their role in making intermediation attractive. Exploring other channels are useful avenues of future research.}

As a fifth and final contribution, I make an observation that the folk theorem--\emph{less information is better in dynamic mechanisms}--does not hold under ex post budget balance and measurability with respect to transfers.\footnote{The folk theorem that less information is better is typically attributed to \citet{myerson1986multi}.} If agents observe their transfers every period, then letting them see each other's reports to the mechanism designer at the end of each period does better in terms of efficiency than hiding this information.\\

{\bf Related Literature.} The bilateral trading problem I study has a rich tradition in static mechanism design---\citet{cs1983bargaining} and \citet{myerson1983efficient} being two of the early papers. \citet{cs1983bargaining} looked at the second-best, in the process defining the solving for the double auction game. \citet{myerson1983efficient} impose all the institutional details on the model and then establish the impossibility of achieving the first-best or efficient allocation. On the other hand, \citet{mm1994bayesian}, \citet{williams1999eff}, and \citet{krishna2000eff} fix an efficient, incentive compatible and individually rational mechanism and show that it can never satisfy budget balance, thereby again proving the same impossibility result using a different technique. I mainly exploit a dynamic generalization of this latter approach.

This paper belongs to the rapidly developing literature on dynamic mechanism design.\footnote{Some of the early papers include \citet{courty2000sequential}, \citet{battaglini2005long}, and \citet{esHo2007optimal}. See \citet{vohra2012dynamic}, \citet{pavan2016dynamic}, and \citet{bergemann2018dynamic} for excellent surveys.} Therein the repeated bilateral trade problem provides an able tool to think formally about dynamic incentive and feasibility constraints as approximations for real institutions. A big picture question in the literature is: What are the apt generalizations of the possibility or impossibility of efficiency results from static mechanism design?

The most well known static mechanisms that implement the efficient allocation are the Vickery-Clarkes-Groves (VCG) and d'Aspremont and G\'{e}rard-Varet (AGV) mechanisms. \citet{bergemann2010dynamic} and \citet{athey2013efficient} develop their respective dynamic generalizations. \citet{bergemann2010dynamic} directly construct a dynamic version of the VCG mechanism that satisfies an efficient exit condition. They also provide an intuitive dynamic auction implementation. Their mechanism, however, does not satisfy budget balance. \citet{athey2013efficient} construct the appropriate generalization of the AGV mechanism to the dynamic environment. It satisfies ex post budget balance. However, they do not allow for individual rationality in every period, thereby demanding a very strong form of commitment from the agents. When they do allow the agents to walk away, a folk theorem is established for the underlying stochastic game.

This paper is the most closely related to \citet{athey2007efficiency} and \citet{skrzypacz2015mechanisms}. \citet{athey2007efficiency} study the repeated bilateral trading problem under iid types, ex ante and ex post budget  balance, and ex post incentive compatibility. They use a bounded budget account to show approximate efficiency under ex post budget balance. \citet{skrzypacz2015mechanisms} analyze the same problem with persistent types and multidimensional initial information. They establish a necessary and sufficient condition for efficiency under ex ante budget balance, thereby requiring high levels of commitment from a third party willing to subsidize trade.\footnote{Using the balancing trick of \citet{athey2013efficient}, this condition also guarantees implementation under ex post budget balance, but then like Athey and Segal, individual rationality has to be relaxed beyond the first period.}

In contemporaneous work, \citet{yoon2015bb} establishes a limit possibility result under ex ante budget balance. He shows that if the agents' types follow an irreducible Markov process, there always exits a high enough discount factor that achieves efficiency.

I depart from previous work in two important ways. First, I demand efficiency under interim and ex post budget balance, while maintaining individual rationality for both agents in each period. This creates a repeated Myerson and Satterthwaite like situation where the three constraints interact every period. Therefore, dynamics is the primary driver of possibility results. Second, I do not just look for limit results. I characterize the entire set of parameters for which efficiency can be sustained.

I also establish a dynamic generalization of the payoff equivalence result. Previous work has stated and proven ex ante versions of the payoff equivalence result for dynamic models, see for example \citet*{pavan2014dynamic}.\footnote{The payoff equivalence result is proven by what has come to be known as the dynamic envelope formula. See also \citet{battaglini2015optimal} and \citet{szentes2015dynamic}.} These results establish that any two incentive compatible mechanisms implementing the same allocation differ in the ex ante net present value of transfers to an agent through a constant. But what can we say about the equivalence of payoffs at other points of history? I prove that the payoffs for any two mechanisms implementing the same allocation are connected by history dependent constants in each period of the dynamic mechanism. This result is used for proving the equivalence of implementation under interim and ex post budget balance, and to characterize the class of all incentive compatible and individually rational mechanisms that implement the efficient allocation.

None of the aforementioned papers provide a intuitively plausible implementation rule of the efficient allocation under budget balance. And, neither do they capture the magnitude of transfers as a function of the parameters. This paper provides a very simple mechanism of implementing the efficient mechanism with the help of an intermediary. Moreover, using the two types model, it also provides a sense of the magnitude of transfers made each period.

Finally, there is a rich literature in financial economics on the desirability of intermediation. \citet{allen2001finint}, and \citet{gorton2002handbook} are some references. To the best of my knowledge, there has been no theoretical attempt at understanding how intermediation can actually increase welfare while ensuring participatory rents for the intermediary through a long-term contract. In an admittedly stylized setup, this paper shows how the presence of an intermediary can increase the efficiency of trade under informational constraints.

\section{Model}

Two agents, each with private information, agree to be in a dynamic bilateral trading relationship for a non-durable good. The buyer (B) has a hidden valuation for the good and the seller (S) is endowed with a technology to produce the good each period at a hidden cost. I assume that the buyer's valuation and the seller's cost are random variables denoted by $v$ and $c$.\footnote{These shall interchangeably be referred to as their types.} They are distributed according to priors $F$ and $G$ on $\mathcal{V} = \left\{v_{1},...., v_{N}\right\}$ and $\mathcal{C}=\left\{c_{1},...., c_{M}\right\}$, and evolve according to independent Markov processes $F(.\vert.): \mathcal{V}\times\mathcal{V}\rightarrow [0,1]$ and $G(.\vert.):\mathcal{C}\times\mathcal{C}\rightarrow [0,1]$, respectively. Both $F(.\vert .)$ and $G(.\vert .)$ satisfy first order stochastic dominance.\footnote{Mathematically speaking $F(v'\vert v)$ is weakly decreasing in $v$ and $G(c'\vert c)$ is weakly increasing in $c$.}
The densities have full support and are denoted by $f$, $g$, $f(.\vert.)$ and $g(.\vert.)$, respectively. Further, $\mathcal{V}\cap \mathcal{C} = \emptyset$, $v_{i+1}>v_{i}$, $c_{j+1}>c_{j}$, and denote $\Delta v_{i+1} = v_{i+1} - v_{i}$, $\Delta c_{j+1} = c_{j+1} - c_{j}$. For an intuitive exposition, I will often write $\underline{v} = v_{1}$ and $\overline{c} = c_{M}$.\footnote{I choose the discrete type model for three reasons. First, it elucidates the key economic forces without having to keep track of measure theoretic details. Second, it allows us to do comparative statics for simple examples, specifically when both agents have two possible types. And, finally many applications of dynamic mechanism design use numerical methods which require a discrete state space. All results naturally converge to their appropriate continuous types counterparts, see \citet{lamba2014phd}.}

Each period $p_{t}$ determines the probability of trade, that is, the production and allocation of the good from the seller to the buyer. $x_{B,t}$ denotes the transfer from the buyer to the mechanism designer, and $x_{S,t}$ the transfer to the seller from the mechanism designer. The per period payoffs are given by $v_{t}p_{t} - x_{B,t}$ and $x_{S,t} - c_{t}p_{t}$ for the buyer and seller respectively.\footnote{The $t$ subscript will not be used when the set of histories make the time dimension obvious.} The (contractual) relationship lasts for $T$ discrete periods, where $T\leq \infty$. Both the agents discount future payoffs with a common discount factor $\delta$. \citet{myerson1983efficient}, and \citet{cs1983bargaining} studied the static version of this model, $\delta = 0$, with continuous type spaces.

In addition to being a mediator who receives reports and recommends actions in the classical Myersonian sense, the mechanism designer here can be considered as a financial intermediary, an institution as part of a larger social contract facilitating trade, or a simple transfer scheme in case $x_{B}= x_{S}$. I will refer to this economic agent interchangeably as the mediator, mechanism designer or intermediary.\footnote{\citet{myerson1986multi} would refer to the mediator as a fictitious agent separate from the intermediary.}

Taking the institutional details as given, both the buyer and seller can commit to the mechanism. Participation constraints every period temper the role commitment will play in the model, as we elaborate below. Both agents know their first period valuation and cost respectively when the contract is signed, and these then stochastically evolve over time. All the parametric and institutional details enunciated thus far are common knowledge between the buyer and seller.

In the spirit of much of the literature, I look at dynamic (history-dependent) direct mechanisms. Every period the agents learn their own types, and then send a report to the mechanism designer, who in turn specifies the allocation (that is, the probability of trade)  and transfer rules. Unlike the static model, the amount of information revealed by the mechanism to the agents matters for future incentives (see \citet{myerson1986multi}). I will focus on the \emph{public mechanism} where one agent's report is observed by the other at the end of each period.\footnote{In Section \ref{sec_hi} I shall introduce and explore the the implications of what I call the \emph{intermediate mechanism}; intermediate because it reveals the hard information of allocations and individual transfers, but hides the soft information of reports and transfers of the other agent. Each agent only observes whether or not trade happens, and transfers are measurable with respect to this information structure. The extreme case would be the \emph{opaque mechanism} where the buyer gets zero information about the seller's reports, and vice-versa. This would of course entail the agents not being able to see allocations or transfers, making it a bit unrealistic for most practical applications.}

Formally, a direct (public) mechanism, say $m$, is a collection of history dependent probability and transfer vectors: $m = \left\langle {\bf p,x}\right\rangle =\left( p\left( v _{t},c_{t}\left\vert h^{t-1}\right.\right),x_{B}\left( v _{t}, c_{t}\left\vert h^{t-1}\right. \right), x_{S}\left( v _{t}, c_{t}\left\vert h^{t-1}\right. \right) \right)_{t=1}^{T}$, where $h^{t-1}$ and $\left(v _{t},c_{t}\right)$ are, respectively, the history of reports up to $t-1$ and current report at time $t$. These can also be succinctly written as $p(h^{t})$, etc. In general, $h^{t}$ is defined recursively as $h^{t}=\left\{ h^{t-1},\left(v _{t},c_{t}\right)\right\} $, with $h^{0}=\emptyset $. I will often write $h^{t} = (v^{t},c^{t})$, where $v^{t}$ and $c^{t}$ are the history of reports by the buyer and seller respectively. The set of possible histories at time $t$ is denoted by $H^{t}$ (for simplicity $H=H^{T} $).

The strategies of the buyer and the seller depend on their private histories. For the buyer, the information available before his period $t$ report is given by $h_{B}^{t} = \left\{h^{t-1}, \hat{v}^{t}\right\}$, starting with $h^{1}_{B} = \left\{\hat{v}_{1}\right\}$, where $\hat{v}^{t} = \left(\hat{v}_{1},...,\hat{v}_{t}\right)$ is the vector of true (actual) types of the buyer. The seller's information is defined analogously. Let the set of these private histories at time $t$ be denoted by $H_{B}^{t}$ and $H_{S}^{t}$, respectively. Thus, fixing the set of parameters: $\Gamma =\left\langle \mathcal{V}\times\mathcal{C}, F,F(.\vert .),G,G(.\vert .),\delta \right\rangle$, for any given mechanism, we have a dynamic game described by $\left\langle m\right\rangle_{\Gamma}$ in which the strategy for the buyer, $\left(\sigma_{B,t}\right)_{t=1}^{T}$, is then simply a function that maps private history into an announcement every period, $\sigma_{B,t}: H^{t}_{B} \mapsto \mathcal{V}$, and similarly for the seller, $\sigma_{S,t}: H^{t}_{S} \mapsto \mathcal{C}$.

\section{The institutional framework}

The edifice of the institutional machinery has three key foundations: private information, voluntary participation and limits on external subisdy. In the mechanism design lexicon, these would respectively be associated with incentive compatibility, individual rationality and budget balance constraints.

For a fixed mechanism $m$, and strategies $\sigma = \left(\sigma_{B}, \sigma_{S}\right)$, the expected utilities on the induced allocation and transfers, after each possible history are defined as follows.
\begin{equation}
\label{eqUb}
U_{B}^{m,\sigma}(h^{t}_{B}) = \mathbb{E}^{m,\sigma}\left[\sum\limits_{\tau=t}^{T}\delta^{\tau-1}\left(v_{\tau}p_{\tau} - x_{B,\tau}\right)\vert h^{t}_{B}\right]
\end{equation}
and,
\begin{equation}
\label{eqUs}
U_{S}^{m,\sigma}(h^{t}_{S}) = \mathbb{E}^{m,\sigma}\left[\sum\limits_{\tau=t}^{T}\delta^{\tau-1}\left(x_{S,\tau}- c_{\tau}p_{\tau}\right)\vert h^{t}_{S}\right]
\end{equation}
Let $U_{i}^{m} = U_{i}^{m\sigma^{*}}$ be the expected utility vector under the truth-telling strategy $\sigma^{*}$. At time $t$, under truth-telling, I will write $U_{B}(v_{t},c_{t}\vert h^{t-1})$ and $U_{S}(v_{t},c_{t}\vert h^{t-1})$ to denote within period ex post expected utility of the buyer and seller respectively.

\subsection{A change of variables}

Much in the spirit of static contract theory\footnote{See for example \cite{stole2001lectures}.}, I propose a change of variables in the structure of the mechanism from $\left\langle {\bf p,x}\right\rangle$ to $\left\langle {\bf p,U}\right\rangle$. As I will later show, this transformation will be central in my endeavor to establish a tight characterization of efficiency.

In order to keep notation simple the type/variable over which expectation is taken is removed. For example
\[p(v_{t}\vert h^{t-1}) = \sum\limits_{j=1}^{M} p\left(v _{t},c_{j,t}\left\vert h^{t-1}\right.\right) g(c_{j,t}\vert c_{t-1}),\]
where $c_{t-1}$ is the $t-1$ period announcement of the seller, known to the buyer.\footnote{$v_{i,t}$ and $c_{j,t}$ refer to type $v_{i}$ for the buyer and type $c_{j}$ for the seller in period $t$.}

Expected utility of the buyer can be recursively defined as
\begin{equation}
\label{eqxU}
U_{B}(v_{t},c_{t}\vert h^{t-1}) = v_{t}p(v_{t},c_{t}\vert h^{t-1})  - x_{B}(v_{t},c_{t}\vert h^{t-1})  + \delta \sum\limits_{i=1}^{N}U_{B}(v_{i,t+1}\vert h^{t-1}, v_{t},c_{t}) f(v_{i,t+1}\vert v_{t})
\end{equation}
Utility of the buyer of type $v_{t}$ from misreporting (once) to be type $v'_{t}$, for a fixed type $c_{t}$ of the seller, can be succinctly written as
\[
U_{B}(v'_{t}c_{t};v_{t}\vert h^{t-1}) =U_{B}(v'_{t},c_{t}\vert h^{t-1}) + (v_{t}-v'_{t})p(v'_{t},c_{t}\vert h^{t-1})  +\]
\[
\delta \sum\limits_{i=1}^{N} U_{B}(v_{t+1,i}\vert h^{t-1},v'_{t},c_{t})\cdot\left(f(v_{i,t+1}\vert v_{t}) - f(v_{i,t+1}\vert v'_{t})\right)
\]
$U_{B}(v_{t}\vert h^{t-1})$ and $U_{B}(v'_{t};v_{t}\vert h^{t-1})$ are defined by taking expectation over $c_{t}$, given $c_{t-1}$.\footnote{That is, $U_{B}(v_{t}\vert h^{t-1}) = v_{t}p(v_{t}\vert h^{t-1})  - x_{B}(v_{t}\vert h^{t-1})  + \delta \sum\limits_{i=1}^{N}U_{B}(v_{i,t+1}\vert h^{t-1}, v_{t}) f(v_{i,t+1}\vert v_{t})$, and $U_{B}(v'_{t};v_{t}\vert h^{t-1})  =U_{B}(v'_{t}\vert h^{t-1}) + (v_{t}-v'_{t})p(v'_{t}\vert h^{t-1})  +\delta \sum\limits_{i=1}^{N}  U_{B}(v_{i,t+1}\vert h^{t-1},v'_{t})\cdot\left(f(v_{i,t+1}\vert v_{t}) - f(v_{i,t+1}\vert v'_{t})\right)$. } The seller's utility, $U_{S}$, can be similarly defined.

It is straightforward to note that a mechanism $m =\left\langle {\bf p,x}\right\rangle$, which is a collection of history dependent allocation and transfer vectors, can be equivalently defined to be $m =\left\langle {\bf p,U}\right\rangle$, where (fixing the allocation) the duality between transfers and expected utility vectors is completely described by equation (\ref{eqxU}).

\subsection{Incentive compatibility}

Exploiting the one-shot deviation principle, incentive compatibility can be defined as follows.
\begin{definition}
A mechanism $m=\left\langle {\bf p,U}\right\rangle$ satisfies incentive compatibility if
\[U_{B}(v_{t}\vert h^{t-1}) \geq U_{B}(v'_{t}; v_{t}\vert h^{t-1})\quad \text{and} \quad U_{S}(c_{t}\vert h^{t-1}) \geq U_{S}(c'_{t}; c_{t}\vert h^{t-1})\]
$\forall v_{t}, v'_{t} \in \mathcal{V}$, $\forall c_{t}, c'_{t} \in \mathcal{C}$, $\forall h^{t-1} \in H^{t-1}$, $\forall t$.
\end{definition}

It states that along all truthful histories, buyer and the seller have no incentive to misreport their type. In order to define a perfect Bayesian equilibrium in the underlying dynamic game associated with this dynamic mechanism design problem, we assume that both agents believe that the other agent reports truthfully with probability one.\footnote{Note that the full support assumption of the Markov processes governing evolution of types for the buyer and seller plays a key role here. Suppose the seller has two possible types $\left\{c_{H},c_{L}\right\}$, and type $c_{L}$ is an absorbing state. Then, if a type $c_{H}$ misreports to be type $c_{L}$, the beliefs which put probability one on truthtelling are no longer consistent, and one-shot deviation property can no longer be applied. Following this report, it is confirmed that the seller would always be type $c_{L}$.} Since agents never observe the other agent's actual type, these beliefs under the above definition of incentive compatibility define a perfect Bayesian equilibrium in the dynamic game.\footnote{See Definition 2 in \citet*{pavan2014dynamic}.}

\subsection{Individual rationality}

Even though  commitment is assumed as part of our institutional architecture, the agents are allowed to walk away after learning their type after any history if their utility from continuing in the contract falls below their reservation thresholds, which are normalized to zero.
\begin{definition}
A mechanism $m=\left\langle {\bf p,U}\right\rangle$ satisfies individually rationality if
\[U_{B}(v_{t}\vert h^{t-1}) \geq 0 \quad \text{and} \quad U_{S}(c_{t}\vert h^{t-1}) \geq 0\]
$\forall v_{t}\in \mathcal{V}$, $\forall c_{t} \in \mathcal{C}$, $\forall h^{t-1} \in H^{t-1}$, $\forall t$.
\end{definition}
A mechanism is termed \emph{implementable} if it is incentive compatible and individually rational.

\subsection{Budget balance}

In mechanism design budget balance is seen as the limits on insurance or external subsidies available to the agents. In addition to the traditional notions of ex ante and ex post budget balance, we introduce an intermediate notion of interim budget balance.

A mechanism is \emph{interim budget balanced} if
\[\mathbb{E}\left[\sum\limits_{\tau = t}^{T}\delta^{\tau-t} \left(x_{B,\tau} - x_{S,\tau}\right) \mid h^{t-1}\right] \geq 0\]
$\forall$ $h^{t-1} \in H^{t-1}$, $\forall t$. The mechanism is \emph{ex ante budget balanced} if interim budget balance holds for the null history. Moreover, the mechanism is \emph{ex post budget balanced} if the entire vector of transfers are are point-wise equal at every history, $x_{B} = x_{S}$.

Using equations (\ref{eqUb}) and (\ref{eqUs}) we can write the expected budget surplus that a mechanism, $m$, generates after any history $h^{t-1}$ to be
\begin{equation}
\label{eqEBS}
\Pi(h^{t-1}) = \mathbb{E}^{m}\left[\sum\limits_{\tau=t}^{T}\delta^{\tau -t}\left(v_{\tau} - c_{\tau}\right)p_{\tau} - U_{B}(v_{t}\vert h^{t-1}) - U_{S}(c_{t}\vert h^{t-1}) \mid h^{t-1}\right]
\end{equation}
The ex ante budget surplus is denoted simply by $\Pi = \Pi(h^{0})$.
\begin{definition}
\label{dIBS}
A mechanism $\left\langle {\bf p,U}\right\rangle$ satisfies interim budget balance if
\[\Pi(h^{t-1}) \geq 0 \quad \forall h^{t-1}\in H^{t-1}, \text{ } \forall t\]
\end{definition}
This can be motivated in many ways. First, it can be viewed as a participation constraint for the mechanism designer---after any history, just like the two agents, the mechanism designer must have an incentive to continue in the relationship. Second, it is a bankruptcy constraint for the intermediary. If the contract reaches a stage the where the intermediary expects to loose money, it should be allowed to shut shop.
\begin{definition}
\label{dEBS}
A mechanism $\left\langle {\bf p,U}\right\rangle$ satisfies ex ante budget balance if
\[\Pi \geq 0\]
\end{definition}
This is the weakest possible notion of budget balance for this dynamic model. It means that the mechanism designer does not loose money in an expected ex ante sense. Importantly, it demands a very strong form of commitment from the intermediary.

Finally, the most standard (and strictest) definition of budget balance from the static literature that can be generalized to dynamic environments states the transfers should exactly equal across all histories for all time periods.
\begin{definition}
A mechanism $m=\left\langle {\bf p,x}\right\rangle$ satisfies ex post budget balance if
\[x_{B}(v_{t},c_{t}\vert h^{t-1}) - x_{S}(v_{t},c_{t}\vert h^{t-1}) = 0,\]
$\forall v_{t}\in \mathcal{V}$, $\forall c_{t} \in \mathcal{C}$, $\forall h^{t-1} \in H^{t-1}$, $\forall t$.
\end{definition}
A natural way to motivate this in the dynamic model is the absence of an outside insurance provider or financial intermediary. Both the agents insure each other against bad shocks, and premium paid can be recovered through continuation utility.

It is fairly easy to see a hierarchy amongst the three notions of budget balance:
\[\text{ex post  budget balance }  \Rightarrow \text{ interim budget balance} \Rightarrow \text{ ex ante budget balance}\]

\subsection{Objective}

One of the most widely accepted objectives of mechanism design is that of efficiency.\footnote{See \citet{holmsmyerson1983eff} for the various notions of efficiency.} We shall invoke the strongest possible version in its ex post form.

\begin{definition}
A mechanism $m=\left\langle {\bf p,U}\right\rangle$ satisfies efficiency if
\[p(v_{t},c_{t}\vert h^{t-1}) =   \left \{
\begin{array} {ccc}
 1 &  \text{ if  } v_{t}>c_{t}  \\
 0 & \text{ otherwise}
 \end{array} \right. \]
 $\forall v_{t}\in \mathcal{V}$, $\forall c_{t} \in \mathcal{C}$, $\forall h^{t-1} \in H^{t-1}$, $\forall t$.
\end{definition}
Trade happens with certainty under a positive instantaneous surplus, and does not happen otherwise.
It is a direct generalization of the notion typically used in static models.

\subsection{The interaction of constraints}

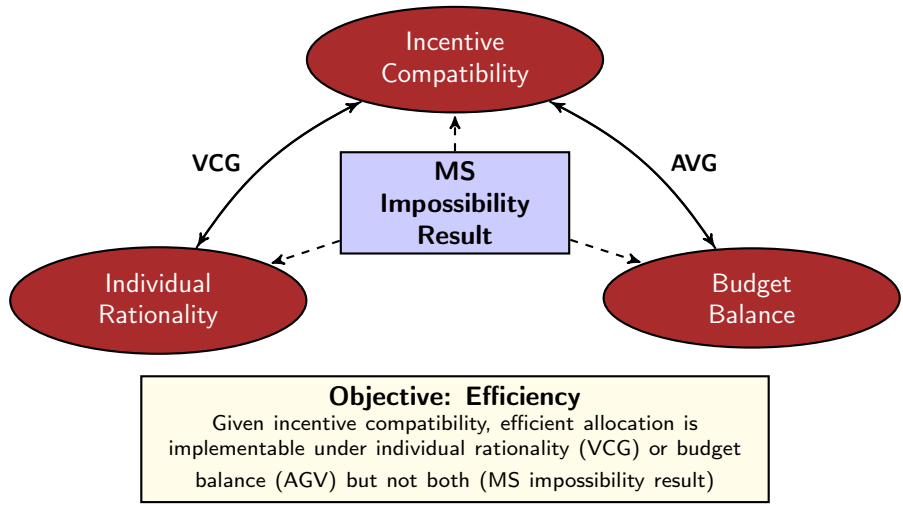
\begin{figure}
\begin{center}
\begin{tikzpicture}[->,>=stealth',shorten >=1pt,auto, align = center,
  thick,main node/.style={draw,font=\sffamily\small}]
\node[main node] (MS) [fill=blue!20, text width=2.75cm] {\textbf{MS \\ Impossibility\\ Result}};
\node[main node] (IC) [ellipse, above = 0.5cm of MS, fill={rgb:red,160;green,40;blue,40}, text width=2.5cm] {\textcolor{white}{Incentive \\ Compatibility}};
\node[main node] (IR) [ellipse, below left = 0.1cm and 1cm of MS,fill={rgb:red,160;green,40;blue,40}, text width=2.5cm] {\textcolor{white}{Individual \\ Rationality}};
\node[main node] (BB) [ellipse, below right = 0.1cm and 1cm of MS, fill={rgb:red,160;green,40;blue,40}, text width=2.5cm] {\textcolor{white}{Budget \\ Balance}};
\node[main node] (Text) [below = 1.6cm of MS, fill = yellow!10, text width=8cm] {\textbf{Objective: Efficiency}\\
\scriptsize{Given incentive compatibility, efficient allocation is \\ implementable under individual rationality (VCG) or budget \\ balance (AGV) but not both (MS impossibility result)}};
\path[every node/.style={font=\sffamily\footnotesize, inner sep=1pt}]
(IC) edge[bend right=15] node[left = 3mm] {\textbf{VCG}} (IR)
(IR) edge[bend left=15] (IC)
(IC) edge[bend left=15] node[right = 3mm] {\textbf{AVG}} (BB)
(BB) edge[bend right=15] (IC)
(MS) edge[dashed] (IC)
(MS) edge[dashed] (IR)
(MS) edge[dashed] (BB);
;
\end{tikzpicture}
\end{center}
\caption{Interaction of constraints in static mechanism design}
\label{fig_interactions}
\end{figure}

Before we jump into the results, it worthwhile to investigate the coexistence of the various forces laid out in this section. From static mechanism design we know that it is the simultaneous interaction of private information (and hence incentive compatibility) with participation (hence individual rationality) and budget balance that leads to the impossibility of efficiency result of \citet{myerson1983efficient}. If we wished to implement the efficient allocation under only individual rationality, the well known Vickery-Clarkes-Groves (VCG) mechanism does so. On the other hand, if we wished to implement the efficient allocation under only budget balance, the d'Aspremont and G\'{e}rard-Varet (AGV) mechanism does so. The three forces all together lead to a departure from efficiency. Figure \ref{fig_interactions} elucidates the three scenarios.

In order to understand how dynamics can overcome the impossibility results, we must model the simultaneous interactions of these forces every period. This cannot be captured under ex-ante budget balance, nor when individual rationality constraint is imposed only in the first period. In this paper I look at ex post and interim budget balance along with repeated individual rationality, which creates a Myerson and Satterthwaite like situation in every period. Herein, dynamics are the primary drivers of possibility results.

\section{Characterization of ex post budget balance}

A standard result in static mechanism design with quasi-linear preferences states that if there exists a mechanism that implements an allocation under ex ante budget balance, then there also exists a mechanism that implements the same allocation under ex post budget balance.\footnote{See \citet{bn2009budget}. Also, note that the converse of this result is obvious.} The intuition for this result is simple. Aggregate payment is equal to total surplus minus information rents. Ex ante budget balance ensures that the value of the underlying economic relationship is large enough to provide information rents to all types (of all agents). Given such a mechanism, we can redistribute the surplus in a way that payments add up to zero ex post. This is because there are enough resources to compensate for the difference between realized and expected payments by cross subsidizing across types.

Is there an equivalent result in mechanism design for dynamic models? The answer is yes, and I characterize this result for repeated bilateral trade. The key conceptual innovation to this end is the idea of interim budget balance.

\begin{proposition}
\label{p_ante_post}
There exists an implementable mechanism $\left\langle {\bf p,x}\right\rangle$ $\left(\text{or } \left\langle {\bf p,U}\right\rangle\right)$ that satisfies ex post budget balance if and only if there exists another implementable mechanism $\left\langle {\bf p,\hat{x}}\right\rangle$ $\left(\text{or }\left\langle {\bf p,\hat{U}}\right\rangle\right)$ that satisfies interim budget balance.
\end{proposition}

The proof for the static result is constructive, and the proof for Proposition \ref{p_ante_post} is no different. However, the main challenge is in preserving incentive compatibility as we translate an interim budget balanced mechanism to satisfy ex post budget balance. I start with a mechanism that is incentive compatible and individually rational and calibrate it to to satisfy ex post budget balance while preserving the first two properties. \citet{athey2007effbt}, on the other hand, start with a mechanism that is incentive compatible and (ex post) budget balanced, and calibrate it to satisfy individual rationality; however, as described above, they run into the problem of preserving incentives while moving transfers and produce only a partial characterization. Commenting on this, they write:

\begin{quote}
\small{The only degree of freedom the transfers offer in transferring utility across players is a fixed constant K (if it varied with history, it would affect incentives). Thus, we look for a K that allows IR to be satisfied in the worst-case scenario for the buyer, and separately in the worst-case scenario for the seller.}
\end{quote}

The problem emanates from the interlinking of incentives. If we change transfers in period $t$ it affects incentives in period $t-k$ and $t+k$. I get around this problem by operating in the $\left\langle {\bf p,U}\right\rangle$ rather than $\left\langle {\bf p,x}\right\rangle$ space. The following lemma characterizes the class of incentive compatible mechanisms allowing us to sidestep the problem stated by \citet{athey2007effbt}.

\begin{lemma}
\label{l_pe}
If $\left\langle {\bf p,U}\right\rangle$ is an incentive compatible mechanism, and $\left\langle {\bf p,\tilde{U}}\right\rangle$ is another mechanism such that
\[\tilde{U}_{B}(v_{t}\vert h^{t-1}) = U_{B}(v_{t}\vert h^{t-1})  + a_{B}(h^{t-1}) \text{ }\forall v_{t},  \text{ and}\]
\[ \tilde{U}_{S}(c_{t}\vert h^{t-1})= U_{S}(c_{t}\vert h^{t-1})  + a_{S}(h^{t-1}) \text{ }\forall c_{t}\]
for a family of constants $\left(a_{B}(h^{t-1}), a_{S}(h^{t-1})\right)$ where $\vert a_{B}(h^{t-1})\vert, \vert a_{S}(h^{t-1})\vert \leq C$ for some positive constant $C$, then $\left\langle {\bf p,\tilde{U}}\right\rangle$ is also incentive compatible.
\end{lemma}

Lemma \ref{l_pe} should be seen as a dynamic generalization of the standard payoff equivalence result in mechanism design. It states that starting with an incentive compatible contract, we can construct another incentive compatible contract that implements the same allocation by linearly translating the expected utility vectors after every history. It is important to note that $a_{B}$ is independent of $v_{t}$ and $a_{S}$ is independent of $c_{t}$.\footnote{For the continuous types model, we can include the converse of Lemma \ref{l_pe}. See \citet{lamba2014phd}.} It may not be obvious at the outset why adding history dependent constants in this fashion keeps the mechanism incentive compatible. In the appendix we show that the proof is fairly straightforward.

Now, given Lemma \ref{l_pe}, the proof of Proposition \ref{p_ante_post} is carried out in two steps. First, define the instantaneous expected budget surplus to be the following.
\[\Pi_{t}(h^{t-1}) = \Pi(h^{t-1}) - \delta \mathbb{E}\left[\Pi(h^{t})\vert h^{t-1}\right] \]
In terms of $\left\langle {\bf p,x}\right\rangle$ representation of the mechanism,
\[\Pi_{t}(h^{t-1}) = \mathbb{E}\left[x_{B}(v_{t},c_{t}\vert h^{t-1}) - x_{S}(v_{t},c_{t}\vert h^{t-1}) \vert h^{t-1}\right]\]

Starting with an implementable mechanism $\left\langle {\bf p,U}\right\rangle$ that is interim budget balanced, choose a new mechanism with constants in Lemma \ref{l_pe} to be $a_{B}(h^{t-1}) = \beta\Pi(h^{t-1})$ and $a_{S}(h^{t-1}) = (1-\beta)\Pi(h^{t-1})$ for any $\beta\in[0,1]$. It is easy to see that in this new mechanism, say $\left\langle {\bf p,\tilde{U}}\right\rangle$, we have $\tilde{\Pi}(h^{t-1}) = 0$ $\forall$ $h^{t-1} \in H^{t-1}$, $\forall$ $t$; ensuring that $\tilde{\Pi}_{t}(h^{t-1}) = 0$ $\forall$ $h^{t-1} \in H^{t-1}$, $\forall$ $t$.

As the second and final step, $\tilde{\Pi}_{t}(h^{t-1}) = 0$ gives us a mechanism that satisfies ``ex ante budget balance in the static sense" for every history. Thus using the standard static mechanism design tools, we can create a third mechanism $\left\langle {\bf p,\hat{U}}\right\rangle$ (or $\left\langle {\bf p,\hat{x}}\right\rangle$) that satisfies ex post budget balance.

\section{Efficiency: a full characterization}
\label{sec_charac}

Now, we ask the main theoretical question: when can efficiency be sustained under ex post and interim budget balance while satisfying incentive compatibility and individual rationality? To this end, I generalize the techniques from static mechanism design developed in \citet{mm1994bayesian}, \citet{williams1999eff}, and \citet{krishna2000eff} to the dynamic model. Start with a repetition of the static VCG mechanism, adjusted for the discreteness of the type space:

\[p^{vcg}(v_{t},c_{t}\vert h^{t-1})  =   \left \{
\begin{array} {ccc}
 1&  \text{ if } v_{t}>c_{t}  \\
 0 & \text{ otherwise}
 \end{array} \right.  \]
\[x^{vcg}_{B}(v_{t},c_{t}\vert h^{t-1}) =   \left \{
\begin{array} {ccc}
 \min\left\{v\vert v>c_{t}\right\} &  \text{ if } v_{t}>c_{t}  \\
 0 & \text{ otherwise}
 \end{array} \right. \]
 \[x^{vcg}_{S}(v_{t},c_{t}\vert h^{t-1}) =   \left \{
\begin{array} {ccc}
 \max\left\{c\vert v_{t}>c\right\} &  \text{ if } v_{t}>c_{t}  \\
 0 & \text{ otherwise}
 \end{array} \right. \]

A standard VCG mechanism would make the buyer pay his externality on the seller, that is the seller's cost, and the make the seller pay her externality on the buyer, that is the buyer's valuation: $x_{B}(v,c) = c \mathbbm{1}_{\left\{v>c\right\}}$ and $x_{S}(v,c) =  v\mathbbm{1}_{\left\{v>c\right\}}$. This mechanism is also incentive compatible. However, for discrete types, gaps between in the type space ensure that it is not the tightest possible mechanism. The mechanism we choose adjusts for gaps and local incentive constraints bind.\footnote{See \citet{maneakos2009discrete} for such a modified VCG mechanism in the static model with discrete types.} It converges to the original VCG mechanism as the model converges to continuous types. Call this mechanism $\left\langle {\bf p^{vcg},x^{vcg}}\right\rangle$ or $\left\langle {\bf p^{vcg},U^{vcg}}\right\rangle$, where given transfers, the expected utility vectors are defined uniquely by equation (\ref{eqxU}).

This mechanism satisfies incentive compatibility and individual rationality, but may violate all notions of budget balance. Trade happens when $v>c$, and payments flowing to the mechanism designer then may be negative. The challenge thus is to construct a new mechanism that satisfies budget balance whenever the primitives of the model allow us to do so without disturbing incentive compatibility and individual rationality.

I take a leaf from \citet{krishna2000eff}'s book, and construct the unique incentive compatible mechanism that binds the individual rationality constraint of the "lowest" type in each period. If the mechanism thus constructed produces an expected budget surplus after every history, we are done, and only if this mechanism produces an expected budget surplus are we done. I call it the \emph{min-max dynamic VCG mechanism}: starting from a VCG mechanism, while preserving incentive compatibility and individual rationality, it produces the minimal possible utility for the agents, and the maximal possible payments to the mechanism designer after every history. \\

\noindent\framebox[1.1\width]{Constructing the Min-Max Dynamic VGC mechanism}\\

\textbf{Step 1.} Start with the dynamic VCG mechanism, $\left\langle {\bf p^{vcg},U^{vcg}}\right\rangle$.  Note that the mechanism is incentive compatible and individually rational; and moreover it is \emph{tight}: local incentive compatibility constraints hold as equalities.\\ \vspace{1mm}

\textbf{Step 2.} Let ${\bf p^{*}}={\bf p^{vcg}}$. Extract all possible agent-surplus in each period, that is maintaining binding local incentive constraints, select ${\bf U^{*}}$  so that $\inf\limits_{v\in \mathcal{V}} U^{*}_{B}(v\vert h^{t-1}) =  0= \inf\limits_{c \in \mathcal{C}} U^{*}_{S}(c\vert h^{t-1})$ for all $v_{t}, c_{t}$ and $h^{t-1}$. According to Lemma \ref{l_pe}, the construction entails a choice of constants: $a_{B}(h^{t-1}) = -\inf\limits_{v\in \mathcal{V}} U^{vcg}_{B}(v\vert h^{t-1})$ and $a_{S}(h^{t-1}) = -\inf\limits_{c \in \mathcal{C}} U^{vcg}_{S}(c\vert h^{t-1})$. Corresponding ex post expected utility vectors:
\begin{equation}
\label{e_mech_B}
U^{*}_{B}(v_{t},c_{t}\vert h^{t-1})= U^{vcg}_{B}(v_{t},c_{t}\vert h^{t-1})-\inf\limits_{v\in \mathcal{V}} U^{vcg}_{B}(v,c_{t}\vert h^{t-1})
\end{equation}
\begin{equation}
\label{e_mech_S}
U^{*}_{S}(v_{t},c_{t}\vert h^{t-1})= U^{vcg}_{S}(v_{t},c_{t}\vert h^{t-1})-\inf\limits_{c\in \mathcal{C}} U^{vcg}_{S}(v_{t},c\vert h^{t-1})
\end{equation}
for all $v_{t}, c_{t}$ and $h^{t-1}$ complete the description of the mechanism.\\ \vspace{1mm}

\textbf{Step 3.} Show that an incentive compatible and individually rational mechanism guaranteeing efficient trade under interim budget balance can exist if and only if $\left\langle {\bf p^{*},U^{*}}\right\rangle$ runs an expected budget surplus, that is, $\Pi^{*}(h^{t-1}) \geq 0$ $\forall h^{t-1}$, $\forall t$. $\qquad \qquad \qquad \qquad \qquad \qquad \qquad \qquad \qquad \qquad\qquad \qquad\square$ \\

The following result characterizes efficiency in terms of the primitives of the model.

\begin{theorem}
\label{t_main}
There exists an incentive compatible and individually rational mechanism that implements the efficient allocation under interim budget balance (or ex post budget balance) if and only if $\Pi^{*}(h^{t-1})\geq  0$ $\forall h^{t-1}$, $\forall t$.
\end{theorem}

Note that the condition constitutes a set of inequality constraints for the entire set of histories. Since this a necessary and sufficient condition, in principal, there is no hope for simplification. Though two features of the min-max dynamic VCG mechanism make the inequalities easy to digest. First, the within-period ex post expected utility vectors, $\left( U^{vcg}_{B}(v_{t},c_{t}\vert h^{t-1}), U^{vcg}_{S}(v_{t},c_{t}\vert h^{t-1})\right)$ are independent of history, $h^{t-1}$. Thus, the min-max dynamic VCG mechanism defined by equations (\ref{e_mech_B}) and (\ref{e_mech_S}) is also independent of history. Second, interim expected utility depends on history only through expectations. Starting from period 2, conditioning on the type of the other agent in the previous period is the history dependence required, which makes interim expected utility Markov. Moreover, given constant support, it is also stationary. Thus, by construction interim transfers are completely defined by the vector
\begin{equation}
\label{U_star}
\left(U^{*}_{B}(v), U^{*}_{S}(c), U^{*}_{B}(v\vert c), U^{*}_{S}(c\vert v)\right) \quad \forall \text{ } (v,c) \in \mathcal{V}\times \mathcal{C}
\end{equation}
Consequently, in order to check for the possibility of efficiency, we need to evaluate $N\times M +1$ constraints: $\Pi^{*}\geq 0$ and $\Pi^{*}(v,c) \geq 0$ for all $v\in\mathcal{V}$ and $c\in\mathcal{C}$.\footnote{See Section \ref{STP} in the appendix, and Figures \ref{fPi} and \ref{fPivector} in Section \ref{sec_cs} for a fully worked out example where both the buyer and trader have two possible types.}

To explore the formulation of the $\Pi^{*}$-vector in terms of the primitives of the environment, define $\Pi^{vcg}$ and $\Pi^{vcg}(v,c)$ respectively to be the ex ante budget surplus and expected budget surplus after history $(v,c)$ for the dynamic VCG mechanism:
\[\Pi^{vcg} = \mathbb{E}\left[\sum\limits_{t=1}^{\infty} \delta^{t-1}\left((v_{t} -c_{t})\mathbbm{1}_{\left\{v_{t}>c_{t}\right\}}\right) - U^{vcg}_{B}(v_{1},c_{1}) - U_{S}^{vcg}(v_{1},c_{1})\right]\]
\[\Pi^{vcg}(v,c) = \mathbb{E}\left[\sum\limits_{t=2}^{\infty} \delta^{t-2}\left((v_{t} -c_{t})\mathbbm{1}_{\left\{v_{t}>c_{t}\right\}}\right) - U^{vcg}_{B}(v_{2},c_{2}) - U_{S}^{vcg}(v_{2},c_{2}) \mid v_{1} = v, c_{1} = c\right]\]
Given first-order stochastic dominance, the infimum of expected utility vectors invoked in equations (\ref{e_mech_B}) and (\ref{e_mech_S}) is attained at the numerically lowest value for the buyer and the numerically highest for the seller. Thus, it is easy to see that
\begin{equation}
\label{e_pi}
\Pi^{*} = \Pi^{vcg} + U^{vcg}_{B}(\underline{v}) + U^{vcg}_{S}(\overline{c})
\end{equation}
and,
\begin{equation}
\label{e_pi2}
\Pi^{*}(v,c) = \Pi^{vcg}(v,c) + U^{vcg}_{B}(\underline{v}\vert c) + U^{vcg}_{S}(\overline{c}\vert v)
\end{equation}
for all $(v,c) \in \mathcal{V}\times\mathcal{C}$

Equations (\ref{e_pi}) and (\ref{e_pi2}) characterize the budget surplus generated in the minimal dynamic VCG mechanism after every history. In particular, equation (\ref{e_pi}) parallels a similar expression obtained in \citet{skrzypacz2015mechanisms} for the continuous type space model, and their result follows from Theorem \ref{t_main}.

\begin{corollary}
\label{cor_ST}
There exists an incentive compatible and individually rational mechanism that implements the efficient allocation under ex-ante budget balance if and only if $\Pi^{*}\geq 0$.\footnote{In the iid model, ex ante and interim budget balance put the exact same restrictions on the set of implementable allocations. \citet{athey2007efficiency} establish that in a continuous type space iid model with $\mathcal{V}=\mathcal{C}$, efficiency can always be sustained for all $\delta\geq\frac{1}{2}$. The same result holds in the continuous types limit of this model.}
\end{corollary}

The mechanism employed by \citet{skrzypacz2015mechanisms} (and \citet{yoon2015bb}), which I call the \emph{bond mechanism}, extracts all possible surplus from the agents in the first period in the form of an upfront payment or bond, and thus ensures ex ante budget balance. However, my construction performs the surplus extraction in the spirit of \citet{krishna2000eff} after every history, and thus demands smaller payments every period. This can only be done in an incentive compatible manner with the aid of Lemma \ref{l_pe}, and be used as characterization of ex post budget balance through Proposition \ref{p_ante_post}, both of which are stated and established in this paper.

The min-max dynamic VCG mechanism delivers two things: first it precisely characterizes efficiency in this dynamic mechanism design problem, and second it provides a working method of implementing efficiency. I will expound more on the latter point in section \ref{sec_impl}.

\subsection{Comparative statics}
\label{sec_cs}

The necessary and sufficient condition in Theorem \ref{t_main} in terms of the $\Pi^{*}$-vector is a joint restriction on the overlap of support, discounting and stochastic process governing the types. Since \citet{myerson1983efficient} we know that in the continuous types model, impossibility of efficiency result holds given any overlap of support. Moreover, the distance from efficiency (measured by the extent of external subsidy required to restore full efficiency) is a decreasing function of the measure of types for which $v>c$. The same intuition carries over to the dynamic model as well. In this subsection we explore the interaction between persistence of types and level of discounting, and its impact on efficiency.

Note that $\Gamma = \left\{\mathcal{V}, \mathcal{C}, F,G,F(.\vert.),G(.\vert.), \delta\right\}$ defines the set of primitives of the model. Intuitively speaking, a high $\delta$ is better for efficiency, and highly persistent Markov processes of types is bad for efficiency. In a dynamic mechanism with commitment, the iid model has the least informational frictions, whereas the fully persistent (or constant types) case kills any advantages that dynamics may offer. In the next two corollaries we show that for the intermediate case, the possibility of efficiency relies of relative sizes of discounting and persistence.

\begin{corollary}
\label{c_limit_d}
For any $\Gamma\setminus\{\delta\}$, there exists a $\delta^{*}$ such that for all $\delta>\delta^{*}$, the  efficient allocation is implementable under interim (or ex post) budget balance.
\end{corollary}

Corollary \ref{c_limit_d} generalizes the possibility of efficiency result under ex ante budget balance in (contemporaneous work in) \citet{yoon2015bb}. It is also the mechanism design counterpart to the decentralized game studied by \citet{athey2013efficient}. The key take away from this result is that with a compact support, even for arbitrarily high levels of persistence of types, we can always find a discount factor high enough that attains efficiency. Conversely, what about a fixed discount factor and arbitrarily high levels of persistence of types?

To answer this question we need to embed the notion of persistence in some parametrization. Consider the set $\Lambda$ of all Markov processes $F(.\vert.)$ and $G(.\vert.)$ satisfying first-order stochastic dominance, parameterized by $\alpha_{B},\alpha_{S}\in[0,1]$ such that $\lim_{\alpha_{B}\rightarrow 1}f(v\vert v) = 1$ $\forall v\in \mathcal{V}$, and $\lim_{\alpha_{S}\rightarrow 1}g(c\vert c) = 1$ $\forall c\in \mathcal{C}$. Note that $\alpha_{B}$ and $\alpha_{S}$ are just indexes and the set $\Lambda$ imposes fairly minimal restriction on the class of stochastic process. For example, different types could converge to one at differing rates.\footnote{Three simple examples of distributions in $\Lambda$ are as follows. $(i)$ [\emph{Renewal model}] $f(v \vert v) = \alpha_{B}$, $f(v'\vert v) = \frac{1-\alpha_{B}}{N-1}$ for $v\neq v'$, and $g(c\vert c) = \alpha_{S}$, $g(c'\vert c) = \frac{1-\alpha_{S}}{M-1}$ for $c\neq c'$. $(ii)$ [\emph{Mixing with identity}] Let $\tilde{f}$ and $\tilde{g}$ be two Markov matrices that satisfy FOSD. Define $f = (1-\alpha_{B}) \tilde{f} + \alpha_{B}I$, and $g = (1-\alpha_{S}) \tilde{g} + (1-\alpha_{S}) I$, where $I$ is the identity matrix. $(iii)$ [\emph{Truncated Normal}] Let $\Delta v_{i} = \Delta v_{j} = \Delta v$, that is types are equally spaced. Define $f(v'\vert v) = \frac{A_{B}(v)}{\sigma_{B}} \Phi\left(\frac{v'-v}{\sigma^{v}_{B}}\right) \Delta v$, where $\Phi$ is the standard normal, and $A_{B}(v)$ is chosen so that the probabilities sum to 1. In this case $\alpha^{v}_{B} = 1-\sigma^{v}_{B}$, and $\alpha_{B} = \min\limits_{v} \alpha^{v}_{B}$. $g(.\vert.)$ is defined analogously.}

Using this machinery, the following corollary states that even for an arbitrarily high discount factor, we can find a Markov process close enough to the identity matrix so that efficiency cannot be sustained.

\begin{corollary}
\label{c_limit_a}
Suppose $\Pi^{*}(\delta=0)<0$. For any $\Gamma\setminus\left\{F(.\vert.), G(.\vert.)\right\}$, and $\left(F(.\vert.), G(.\vert.)\right) \in \Lambda$, there exists $\alpha^{*}_{B}$ and $\alpha^{*}_{S}$ such that for all $\alpha_{B}>\alpha^{*}_{B}$ and $\alpha_{S}>\alpha^{*}_{S}$, the efficient allocation cannot be implemented under interim (or ex post) budget balance.
\end{corollary}

The two corollaries together formalize perhaps a folk wisdom that holding the (compact) supports fixed, the possibility of efficiency depends on the order of limits between discounting and persistence. This intuition is well captured by the following simple example. Suppose $\mathcal{V} = \left\{v_{H},v_{L}\right\}$, $\mathcal{C} = \left\{c_{H},c_{L}\right\}$, where $v_{H}>c_{H}>v_{L}>c_{L}$.\footnote{Any other arrangement of the two types model will render posted prices efficient in the static model and hence in the repeated model. This was first noted by \citet{matsuo1989twotypes}.} Let $\Pi^{*}(\delta=0)<0$. We call this the simple trading problem (\emph{STP}).\footnote{Note that for the STP, the VCG mechanism (modified for the discrete type space) is given by:
\[x^{vcg}_{B}(v_{H},c_{H}) = v_{H}, x^{vcg}_{B}(v_{H},c_{L}) = v_{L}, x^{vcg}_{B}(v_{L},c_{H}) = 0, x^{vcg}_{B}(v_{L},c_{L}) = v_{L}\]
\[x^{vcg}_{S}(v_{H},c_{H}) = c_{H}, x^{vcg}_{S}(v_{H},c_{L}) = c_{H}, x^{vcg}_{S}(v_{L},c_{H}) = 0, x^{vcg}_{S}(v_{L},c_{L}) = c_{L}\]}
To evaluate of possibility of efficiency wee need to calculate five expressions: $\Pi^{*}$ and $\Pi^{*}(v,c)$ for $(v,c) \in\left\{v_{H},v_{L}\right\}\times\left\{c_{H},c_{L}\right\}$. These can be done in closed form, see Section \ref{STP} in the appendix.

To further simplify the parameter space, inspired by \citet{myerson1985two}, I look at a dynamic version of the symmetric uniform trading problem: $\mathcal{V} = \left\{1,v\right\}$, $\mathcal{C} = \left\{c,0\right\}$ with $1>c>v>0$ and $\Delta v = 1-v = c = \Delta c$; assume a uniform prior $f=g=\left(\frac{1}{2}.\frac{1}{2}\right)$, and symmetric Markov evolution: $f(v\vert v) = g(c\vert c) = \alpha$ for $\frac{1}{2}<\alpha<1$.\footnote{$\Pi^{*}(\delta=0)<0$ simplifies to $c-v>\frac{1}{2}$.} Call this the uniform symmetric simple trading problem (\emph{USSTP}). Owing to symmetry, it immediately follows that the ex post utility vectors for the buyer and seller in this model are equal, that is $U^{*}_{B}=U^{*}_{S}$, and $\Pi^{*}(v_{H},c_{H}) = \Pi^{*}(v_{L},c_{L})$. The latter implies that we need to check four inequality constraints for efficiency.
\begin{figure}
        \centering
        \begin{subfigure}[b]{0.43\textwidth}
                \includegraphics[width=\textwidth]{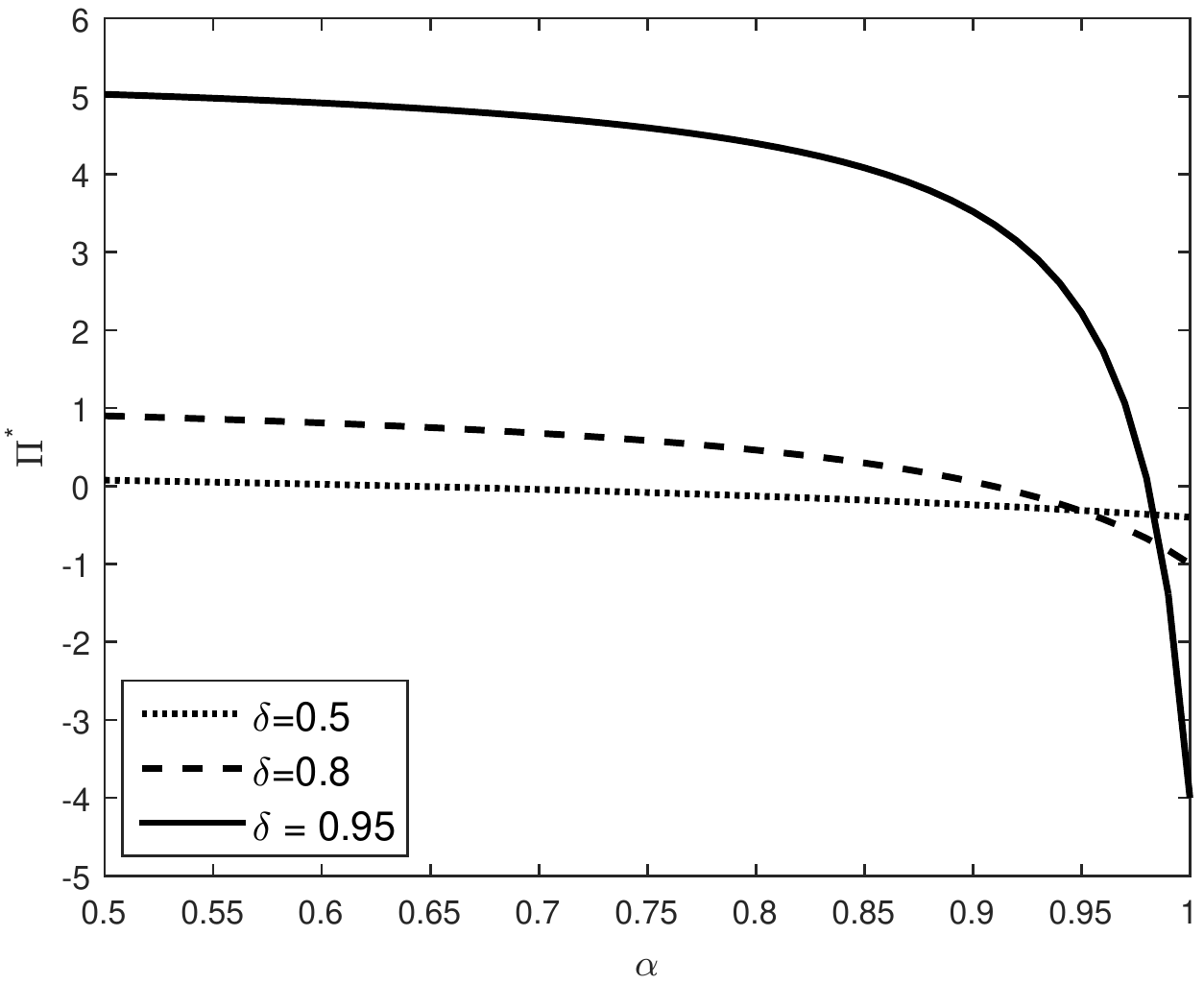}
                \caption{$\Pi^{*}$}
                \label{Pstar}
        \end{subfigure}
        \begin{subfigure}[b]{0.43\textwidth}
                \includegraphics[width=\textwidth]{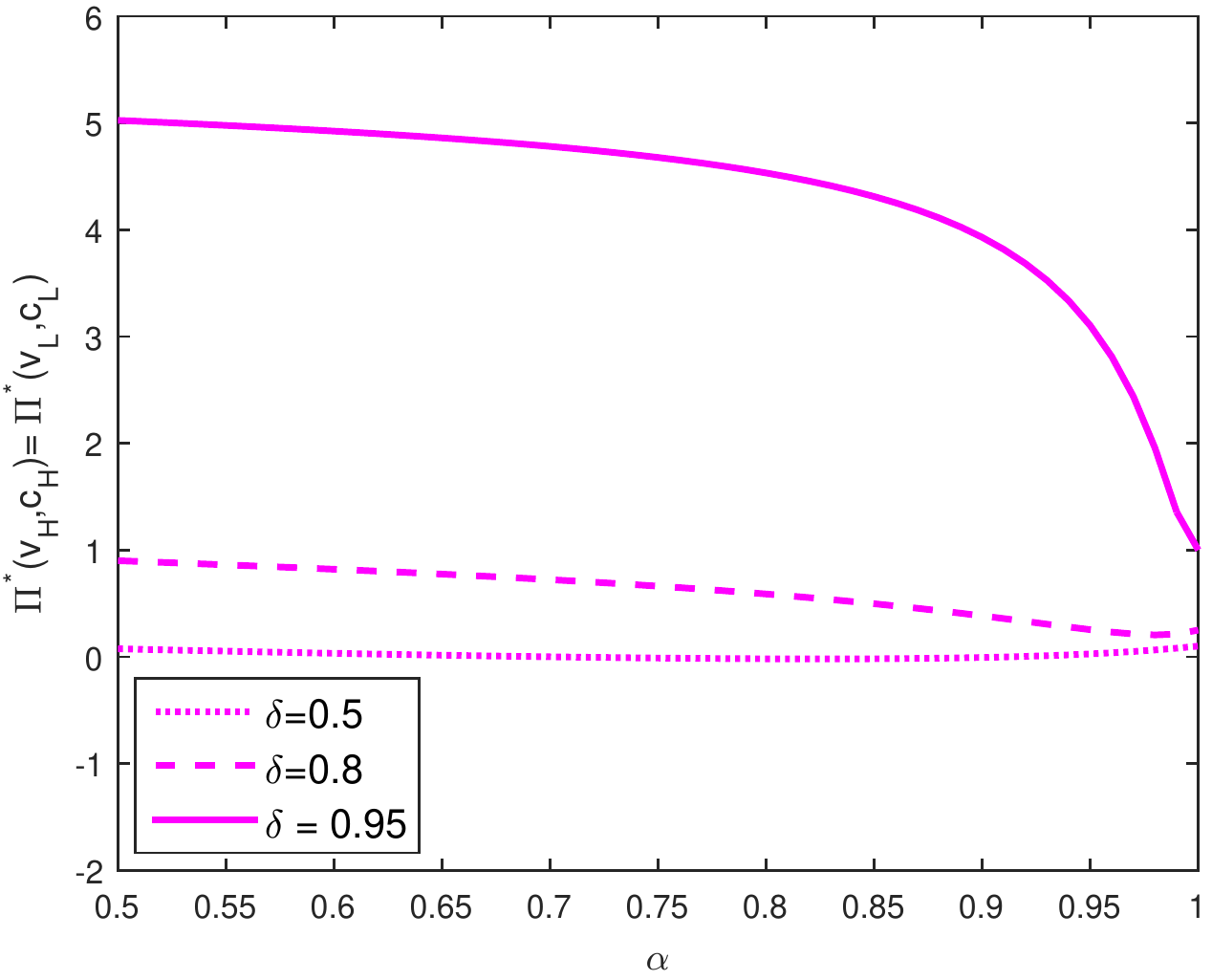}
               \caption{$\Pi^{*}(v_{H},c_{H})=\Pi^{*}(v_{L},c_{L})$}
                \label{Phhstar}
        \end{subfigure} \\
        \begin{subfigure}[b]{0.43\textwidth}
                 \includegraphics[width=\textwidth]{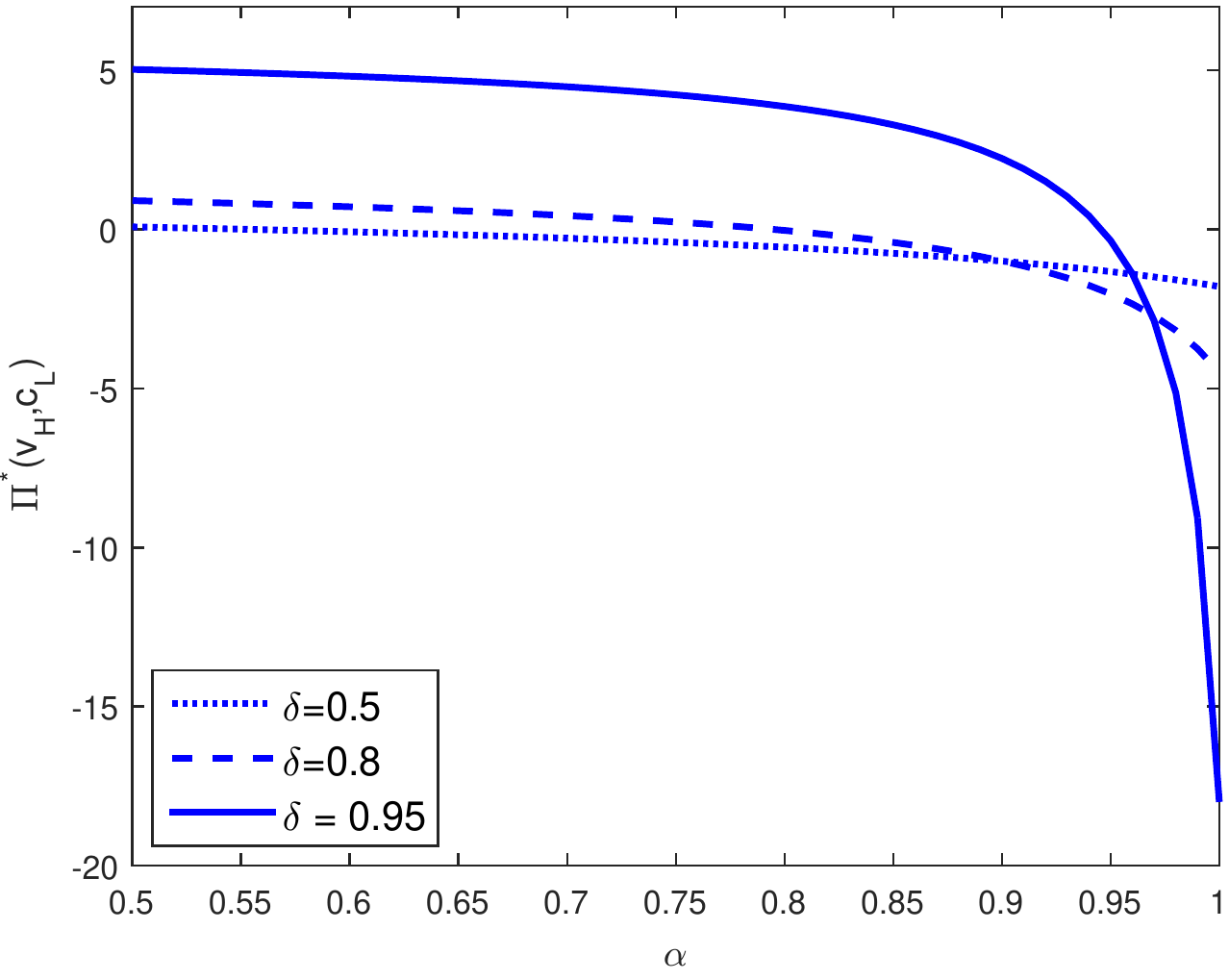}
               \caption{$\Pi^{*}(v_{L},c_{H})$}
                \label{Phlstar}
        \end{subfigure}
          \begin{subfigure}[b]{0.43\textwidth}
                 \includegraphics[width=\textwidth]{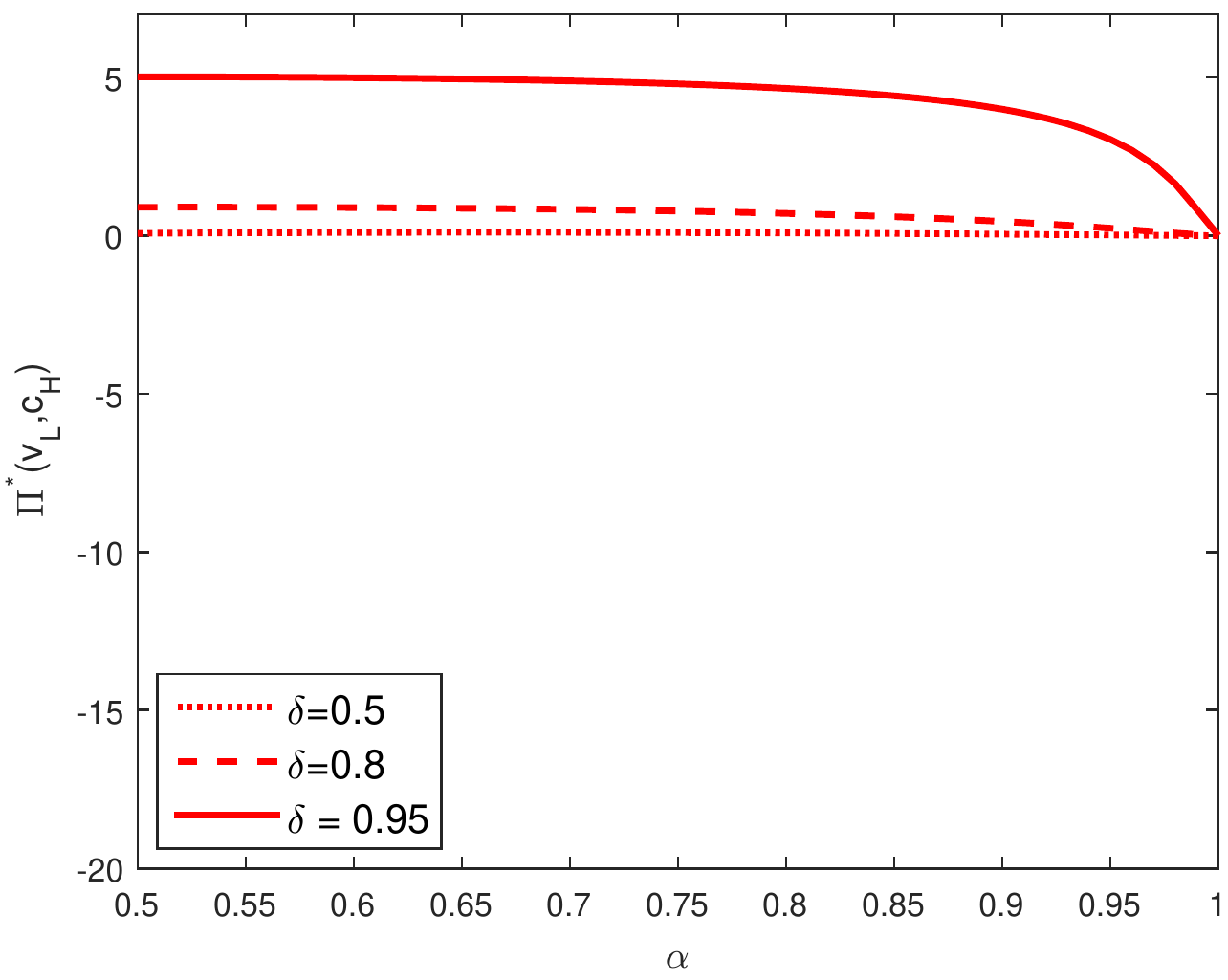}
               \caption{$\Pi^{*}(v_{L},c_{L})$}
                \label{Plhstar}
        \end{subfigure}

       \caption{$\Pi^{*}$ and $\Pi^{*}(v,c)$ for the uniform symmetric simple trading problem, $\delta = 0.5,0.8,0.95$ }\label{fPi}
\end{figure}
\begin{figure}
        \centering
        \begin{subfigure}[b]{0.43\textwidth}
                \includegraphics[width=\textwidth]{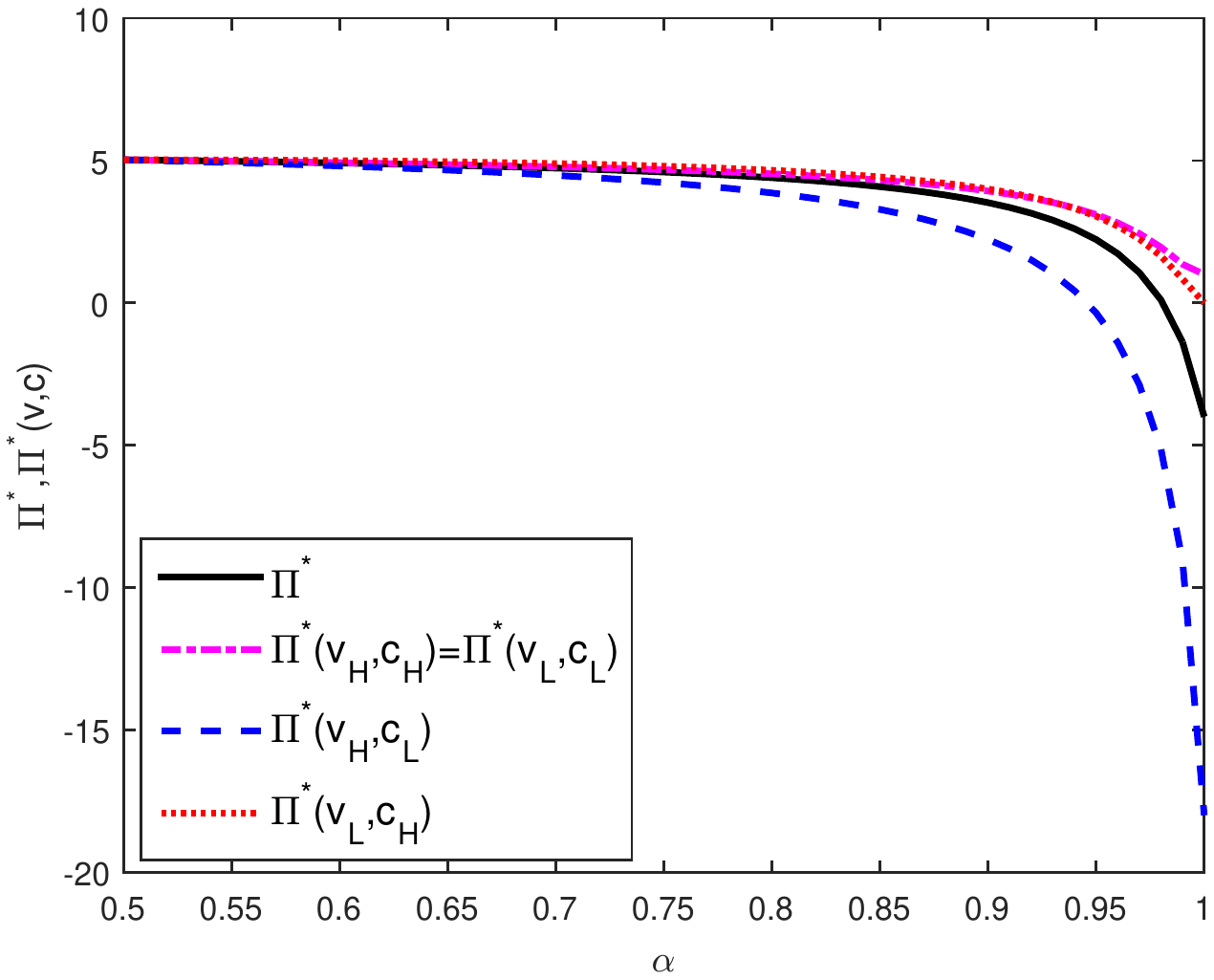}
                \caption{$\delta = 0.95$}
                \label{Piv1}
        \end{subfigure}
        \begin{subfigure}[b]{0.43\textwidth}
                \includegraphics[width=\textwidth]{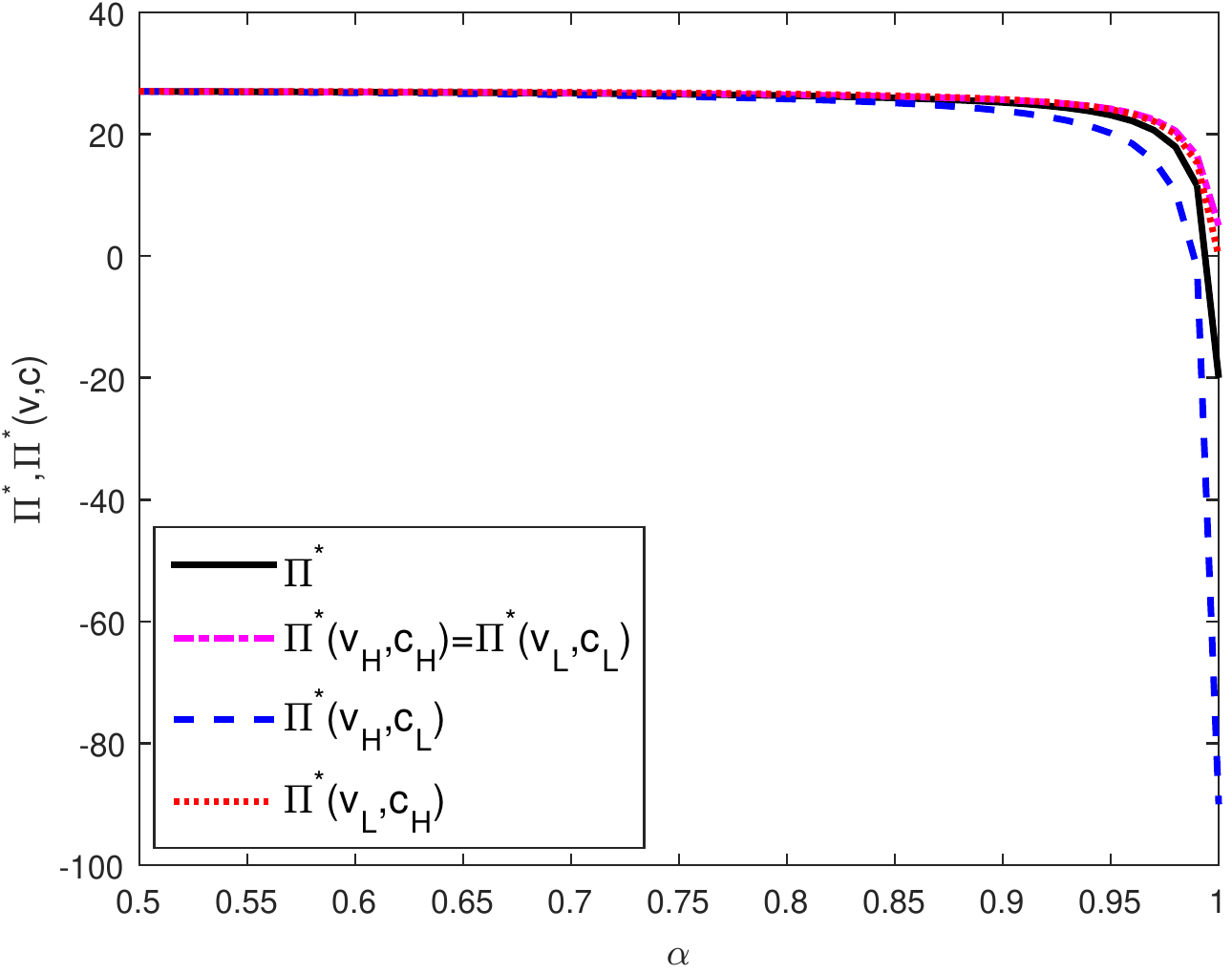}
                \caption{$\delta = 0.99$}
                \label{Piv2}
        \end{subfigure}
       \caption{On the relative "sizes" of $\delta$ and $\alpha$ for the uniform symmetric simple trading problem.}\label{fPivector}
\end{figure}

Figure \ref{fPi} plots each of the $\Pi^{*}$-vector for the \emph{USSTP} as a function of persistence, $\alpha$, for three different values of $\delta$. It is clear that higher persistence is unambiguously bad and higher discount factor is unambiguously good for efficiency, and these hold keeping the other fixed as well. Figure \ref{fPivector} maps all the elements of the $\Pi^{*}$-vector for fixed discount factors 0.95 and 0.99. Three facts must be noted. First, all elements of the vector start at the same value at $\alpha=0.5$, which is akin to the iid model and hence independent of history. Second, there is a clear ranking in the expected budget surplus generated by the efficient allocation. Somewhat counteractively, $\Pi^{*}(v,c)$ is non-monotonic in $v$ and $c$; in particular $\inf_{v,c} \Pi^{*}(v,c) = \Pi^{*}(v_{H},c_{L})$. The best possible types in terms of efficiency, generate the lowest expected budget surplus.\footnote{Equation (\ref{e_pi2}) tells us that $\Pi^{*}(v,c)$ is composed of the budget deficit due the VCG mechanism plus the surplus meted out to the lowest possible type in the VCG mechanism. Now, it can be checked that for the USSTP $\Pi^{vcg}(v_{H},c_{H})= \Pi^{vcg}(v_{L},c_{L})> \Pi^{vcg}(v_{L},c_{H}) > \Pi^{vcg}(v_{H},c_{L})$, $U^{vcg}_{B}(\underline{v}\vert c_{L}) > U^{vcg}_{B}(\underline{v}\vert c_{H})$ and $U^{vcg}_{S}(\overline{c}\vert v_{H}) > U^{vcg}_{S}(\overline{c}\vert v_{L})$. Therefore, even though more surplus is extracted from the agents when the type realization is $(v_{H},c_{L})$ the larger budget deficit of the original VCG mechanism cannot be compensated for enough culminating in the lowest value of in the $\Pi^{*}$-vector.} Third, the scales of y-axis for Figures \ref{Piv1} and \ref{Piv2} are different by some orders of magnitude. At high levels of discounting the values of the $\Pi^{*}$-vector increase non-linearly with the discount factor.

\section{Implementation}
\label{sec_impl}

Given the full characterization result in Section \ref{sec_charac}, a natural question to ask is what plausible mechanisms would implement the efficient allocation. In what follows, constrained by fixing the efficient allocation and the three institutional constraints, I look for simple and intuitively meaningful transfer schemes. The question is divided into implementation under interim and ex post budget balance. Throughout it is assumed that $\Pi^{*}(h^{t-1})\geq  0$ $\forall h^{t-1}$, $\forall t$, or more succinctly, $\Pi^{*}\geq 0$ and $\Pi^{*}(v,c) \geq 0$ for all $v\in\mathcal{V}$ and $c\in\mathcal{C}$.

\subsection{Under interim budget balance}

As pointed out before, the ex post utility vectors in the dynamic min-max VCG mechanism, defined by Equations (\ref{e_mech_B}) and (\ref{e_mech_S}), are forward looking and independent of history. Thus, the ex post transfers in the associated mechanism $\left\langle {\bf p^{*},x^{*}}\right\rangle$ are also history independent, lending a rather simple structure to the direct mechanism.

In this subsection I show, in addition, the min-max dynamic VCG mechanism provides another simple, and perhaps importantly also an exhaustive method of implementing the efficient allocation under interim budget balance. First, I describe the mechanical steps of implementation, and then later expand on its simplicity and exhaustiveness.\\

\noindent\framebox[1.1\width]{Implementing the Min-Max Dynamic VGC mechanism}\\

\textbf{Step 1.} At the start of period $t$, the buyer and seller pay a fixed fee: $z_{B,t}(c_{t-1})$ and $z_{S,t}(v_{t-1})$. From the second period onwards this fee only depends on the type of the other agent in the previous period. Thus, the fee structure is completely described by: $z_{B,1}, z_{S,1}$, and $z_{B}(c), z_{S}(v)$ $\forall (v,c) \in\mathcal{V}\times\mathcal{C}$.  \\ \vspace{1mm}

\textbf{Step 2.} After the fixed fee is paid, the mechanism designer runs a simple static VCG mechanism as defined above: $\left\langle {\bf p^{vcg},x^{vcg}}\right\rangle$. This (stationary) mechanism is repeated every period. \\ \vspace{1mm}

\textbf{Step 3.} Comparing ${\bf U^{vcg}}$ and ${\bf U^{*}}$, by construction the fixed fees are defined by:
\[U^{vcg}_{B}(\underline{v}\vert c) = z_{B}(c) + \delta \sum\limits_{\tilde{c}\in\mathcal{C}}U^{vcg}_{B}(\underline{v}\vert\tilde{c})g(\tilde{c}\vert c), \quad U^{vcg}_{B}(\underline{v}) = z_{B,1} + \delta \sum\limits_{\tilde{c}\in\mathcal{C}}U^{vcg}_{B}(\underline{v}\vert\tilde{c})g(\tilde{c}) \]
\[U^{vcg}_{S}(\overline{c}\vert v) = z_{S}(v) + \delta \sum\limits_{\tilde{v}\in\mathcal{V}}U^{vcg}_{S}(\overline{c}\vert\tilde{v})f(\tilde{v}\vert v), \quad U^{vcg}_{S}(\overline{c}) = z_{S,1} + \delta \sum\limits_{\tilde{v}\in\mathcal{V}}U^{vcg}_{S}(\overline{c}\vert\tilde{v})f(\tilde{v}) \]
$\qquad \qquad \qquad \qquad  \qquad \qquad \qquad \qquad \qquad \qquad \qquad \qquad \qquad \qquad \qquad \qquad \qquad \qquad \qquad \square$ \\

It is easy to the see that the per-period transfers in the original min-max dynamic VCG mechanism constructed in Section \ref{sec_charac} are given by
\[x^{*}_{B}(v,c\vert \tilde{v}, \tilde{c}) = x^{vcg}_{B}(v,c) + z_{B}(\tilde{c}), \quad x^{*}_{B}(v,c) = x^{vcg}_{B}(v,c) + z_{B,1}\]
\[x^{*}_{S}(v,c\vert \tilde{v}, \tilde{c}) = x^{vcg}_{S}(v,c) - z_{S}(\tilde{v}), \quad x^{*}_{S}(v,c) = x^{vcg}_{S}(v,c) - z_{S,1}\]

The mechanism is quite simple: it is Markov in that transfers only depend on the current reports and report of the other agent in the previous period, and it is stationary. It is intuitive in that every period it requires a small participation fees and then runs a VCG mechanism.

The magnitude of the fixed payments with respect to other parameters of the problem are delineated for the uniform symmetric simple trading problem in Table \ref{tab_mag_fixed}. Fixed fee payments have a reasonable magnitude: the never exceed one-fourth of the highest type. Moreover, transfers decrease with the level of persistence. An increase in persistence decreases the amount of total economic surplus, and since the expected utility of the buyer and seller are already at their lowest possible value, this implies a decrease in the share of the pie for the mechanism designer, and hence a reduction in fixed fee payments.
\begin{table}[h!]
\centering
\begin{tabular}{|c|c|c|c|}
 \hline
$\alpha$ &  $z_{B}(c_{H})$ &  $z_{B}(c_{L})$ & $z_{B,1}$  \\ \hline
0.5 & 0.225& 0.225 & 0.225   \\ \hline
0.6 &0.215 & 0.230& 0.222  \\ \hline
0.7 &0.192 & 0.243& 0.218 \\ \hline
0.8 & 0.160& 0.259& 0.209\\ \hline
0.9 & 0.114& 0.261& 0.188\\ \hline
\end{tabular} 
\caption{Fixed fee payments for the \emph{USSTP} in the min-max dynamic VCG mechanism. Parametrization: $\mathcal{V} = \left\{1,0.05\right\},\mathcal{C} = \left\{0.95,0\right\}, \delta = 0.95$.}
\label{tab_mag_fixed}
\end{table}

How is the construction exhaustive in implementability? Consider the class of all possible \emph{tight} mechanisms that implement the efficient allocation, that is all mechanisms where local incentive constraints hold as equalities.\footnote{A formal description is provided in Section \ref{sec_tight} in the appendix.} These mechanisms are the "right" mechanisms to consider in the discrete types framework: they are a natural dynamic analog of the static discrete model, and in the limit of the dynamic continuous type space model they cover all possible incentive compatible mechanisms since local incentive constraints always "bind" owing to the envelope theorem (see \citet*{pavan2014dynamic}).

The min-max dynamic VCG mechanism is \emph{one} such tight mechanism that is uniquely defined by the fact that it transfers all possible expected surplus at every history to the mechanism designer. All other mechanisms are linear translations of this mechanism. A full characterization follows.
\begin{corollary}
\label{cor_all_int_mech}
Suppose $\Pi^{*}(h^{t-1})\geq  0$ $\forall h^{t-1}$, $\forall t$. A tight mechanism $\left\langle {\bf p^{*},U^{\beta}}\right\rangle$ implements the efficient allocation under interim budget balance if and only if
\[U^{\beta}_{B}(v_{t}\vert h^{t-1}) = U^{*}_{B}(v_{t}\vert h^{t-1}) + \beta_{B}(h^{t-1})\Pi^{*}(h^{t-1})\]
\[U^{\beta}_{S}(v_{t}\vert h^{t-1}) = U^{*}_{S}(c_{t}\vert h^{t-1}) + \beta_{S}(h^{t-1})\Pi^{*}(h^{t-1})\]
$\forall$ $h^{t-1}$, $\forall$ $t$ with $\beta_{B}(h^{t-1}), \beta_{S}(h^{t-1})\geq 0$, and $\beta_{B}(h^{t-1}) + \beta_{S}(h^{t-1}) \leq 1$.
\end{corollary}

Note that $\beta_{B}(h^{t-1}), \beta_{S}(h^{t-1})$, and $1-\beta_{B}(h^{t-1})-\beta^{S}(h^{t-1})$ are the fractions of surplus $\Pi^{*}(h^{t-1})$ shared by the buyer, seller and the mechanism designer after history $h^{t-1}$. Just as before this new mechanism can be implemented using a fixed fee and VCG mashup.  And, if we choose $\beta_{B}$ and $\beta_{S}$ to be Markov and stationary, the new mechanism will also adopt the Markov and stationary properties of the original min-max dynamic VCG mechanism.\footnote{For the min-max dynamic VCG mechanism $\beta_{B}=\beta_{S}=0$.}

All possible (history dependent) bargaining protocols that can implement the efficient allocation under interim budget balance are captured by the vector ${\bf \beta} = \left(\beta_{B}(h^{t-1}),\beta_{S}(h^{t-1})\right)$. Once the fact of efficiency is established by the min-max dynamic VCG mechanism, this vector lists all possible divisions of the economic surplus that can emerge in equilibrium.

This a good place to contrast my mechanism from the \emph{bond} mechanisms typically adopted in the literature (see  \citet{skrzypacz2015mechanisms} and \citet{yoon2015bb}) in which the agents are asked to pay an upfront fees ensuring ex ante (but not interim or ex post) budget balance. In Table \ref{table_bond} (for the uniform symmetric simple trading problem), the maximum fixed fee payment across all possible histories in the min-max dynamic VCG mechanism is normalized to 1. The upfront payment is always $1400\%$ more than the fixed fee payment. The bond mechanism, while plausible in certain contexts, may not work unless the agents have very deep pockets; small credit constraints will render it infeasible.
\begin{table}[h!]
\centering
\begin{tabular}{|c|c|c|}
\hline
$\alpha$ &  $\max \vert z_{B} \vert$ (normalized) & $UP$  \\ \hline
0.5 & 1 & 2000\%     \\ \hline
0.6 & 1 & 1934\%   \\ \hline
0.7 & 1 & 1790\%  \\ \hline
0.8 & 1 & 1619\% \\ \hline
0.9 & 1 & 1437\% \\ \hline
\end{tabular}
\caption{Maximum fixed fess in dynamic min-max VCG mechanism versus upfront payment (UP) in the bond mechanism for the \emph{USSTP}. Parametrization: $\mathcal{V} = \left\{1,0.05\right\},\mathcal{C} = \left\{0.95,0\right\}, \delta = 0.95$.}
\label{table_bond}
\end{table}

\subsection{Under ex post budget balance}

Unlike the case with an active mechanism designer (or intermediary), implementation under ex post budget balance does not allow differential payments to be made by the buyer and to the seller. Restriction to "direct transfers" shrinks the family of incentive compatible sharing rules of the future economic surplus, but as established in Proposition \ref{p_ante_post}, it does not change the fact of efficiency.

The direct mechanism for implementation under ex post budget balance has a clear intuitive formulation. Starting from the dynamic min-max VCG mechanism, let $\left\langle {\bf p^{*},U^{o}}\right\rangle$ be the mechanism that generates exactly zero expected budget surplus after every history and distributes any extra economic surplus to both agents equally; that is $\forall$ $h^{t-1}$:
\begin{equation}
\label{ub_zero}
U^{o}_{B}(v_{t}\vert h^{t-1}) = U^{*}_{B}(v_{t}\vert h^{t-1})  + \frac{1}{2}\Pi^{*}(h^{t-1}) \text{ }\forall v_{t}
\end{equation}
\begin{equation}
\label{us_zero}
U^{o}_{S}(c_{t}\vert h^{t-1})= U^{*}_{S}(c_{t}\vert h^{t-1})  + \frac{1}{2}\Pi^{*}(h^{t-1}) \text{ }\forall c_{t}
\end{equation}
The new mechanism is of course Markov and stationary starting period 2. Importantly, it generates a zero static expected budget surplus after any history:
\[\Pi^{o}_{t}(h^{t-1}) =\mathbb{E}\left[x^{o}_{B}(v_{t},c_{t}\vert h^{t-1}) - x^{o}_{S}(v_{t},c_{t}\vert h^{t-1}) \vert h^{t-1}\right]= \Pi^{o}(h^{t-1}) - \delta \mathbb{E}\left[\Pi^{o}(h^{t})\vert h^{t-1}\right] = 0 \]

The direct mechanism implementing the efficient allocation is then completely characterized by the following payments:
\begin{equation}
\label{transfers_expost0}
x^{\dagger}(v,c) = x^{o}_{S}(c) + \left[x^{o}_{B}(v) - \bar{x}^{o}_{B}\right]
\end{equation}
\begin{equation}
\label{transfers_expost}
x^{\dagger}(v,c\vert \tilde{v},\tilde{c}) = x^{o}_{S}(c\vert \tilde{v}, \tilde{c}) + \left[x^{o}_{B}(v\vert \tilde{v},\tilde{c}) - \bar{x}^{o}_{B}(\tilde{v},\tilde{c})\right]
\end{equation}
where $\bar{x}^{o}_{B} = \mathbb{E}\left[x_{B}^{o}(v,c)\right]$, and $\bar{x}^{o}_{B}(\tilde{v},\tilde{c}) = \mathbb{E}\left[x_{B}^{o}(v,c\vert \tilde{v},\tilde{c})\mid (\tilde{v},\tilde{c}) \right]$ are the total expected payments by the buyer in period 1 and period 2 respectively in the mechanism $\left\langle {\bf p^{*},x^{o}}\right\rangle$. Note that $x^{\dagger}(v) = x^{o}_{B}(v)$, $x^{\dagger}(v\vert \tilde{v},\tilde{c}) = x^{o}_{B}(v\vert \tilde{v},\tilde{c})$, and $x^{\dagger}(c) = x^{o}_{S}(c)$, $x^{\dagger}(c\vert \tilde{v},\tilde{c}) = x^{o}_{S}(c\vert \tilde{v},\tilde{c})$. Therefore, the mechanism $\left\langle {\bf p^{*},x^{\dagger}}\right\rangle$, where ${\bf x_{B}^{\dagger} = x_{S}^{\dagger} = x^{\dagger}}$, is by construction implementable and satisfies ex post budget balance.

While the interim budget balance allows for dynamically balancing transfers through external liquidity, under ex post budget balance the promises are direct and have to be fulfilled over time by the agents themselves. The agents take over the role of both the producers of economic surplus and liquidity providers; the two components of the payment scheme $x^{o}_{S}(c\vert \tilde{v}, \tilde{c})$ and $\left[x^{o}_{B}(v\vert \tilde{v},\tilde{c}) - \bar{x}^{o}_{B}(\tilde{v},\tilde{c})\right]$, respectively mirror the two roles.

Table \ref{table_expost} reports payments in the mechanism $\left\langle {\bf p^{*},x^{\dagger}}\right\rangle$ for two possible histories. $x^{\dagger}(v_{H},c_{L}\vert$ $v_{H},c_{L})$ represents the transfer for the most favorable realization of types in terms of trade, and $x^{\dagger}(v_{L},c_{H}\vert v_{L},c_{H})$ the least favorable. It is clear that while the former involves payments from the buyer to the seller, in the latter it is in fact the seller that compensates the buyer for the realization of bad types. The flow of funds in both directions captures the role of agents as liquidity providers in the absence of an intermediary.
\begin{table}[h!]
\centering
\begin{tabular}{|c|c|c|c|c|c|}
\hline
$\alpha$ &  $x^{\dagger}(v_{H},c_{L}\vert v_{H},c_{L}) $ & $x^{\dagger}(v_{H},c_{L}\vert v_{H},c_{H}) $ & & $x^{\dagger}(v_{L},c_{H}\vert v_{L},c_{H})$ & $x^{\dagger}(v_{L},c_{H}\vert v_{H},c_{H})$  \\ \hline
0.5 & 0.625 & 0.625 & & 0.125 & 0.125 \\ \hline
0.6 & 0.596 & 0.742 & & 0.009 & 0.118\\ \hline
0.7 & 0.567 & 0.879 & &-0.096 & 0.043\\ \hline
0.8 & 0.540 & 1.090 & &-0.195 &-0.178\\ \hline
0.9 & 0.517 & 1.607 & &-0.289 &-0.8831\\ \hline
\end{tabular}
\caption{Transfers under ex post budget balance in an equal sharing of economic surplus for the \emph{USSTP}. Parametrization: $\mathcal{V} = \left\{1,0.05\right\},\mathcal{C} = \left\{0.95,0\right\}, \delta = 0.95$.}
\label{table_expost}
\end{table}

Now, what about other possible (direct) mechanisms that implement the efficient allocation under ex post budget balance? A simple characterization result follows.
\begin{corollary}
\label{cor_expost_charac}
Suppose $\Pi^{*}(h^{t-1})\geq  0$ $\forall h^{t-1}$, $\forall t$. If a tight mechanism $\left\langle {\bf p^{*},U^{\beta}}\right\rangle$ implements the efficient allocation under ex post budget balance, then
\[U^{\beta}_{B}(v_{t}\vert h^{t-1}) = U^{*}_{B}(v_{t}\vert h^{t-1}) + \beta(h^{t-1})\Pi^{*}(h^{t-1})\]
\[U^{\beta}_{S}(v_{t}\vert h^{t-1}) = U^{*}_{S}(c_{t}\vert h^{t-1}) + (1-\beta(h^{t-1}))\Pi^{*}(h^{t-1})\]
$\forall$ $h^{t-1}$, $\forall$ $t$, for some $\beta(h^{t-1})\in [0,1]$ $\forall$ $h^{t-1}$, $\forall$ $t$. Conversely, if a mechanism $\left\langle {\bf p^{*},U^{\beta}}\right\rangle$ satisfies these two equations for a family of constants $\beta(h^{t-1})\in [0,1]$ $\forall$ $h^{t-1}$, $\forall$ $t$, then there exists a mechanism $\left\langle {\bf p^{*},\tilde{U}^{\beta}}\right\rangle$ that implements the efficient allocation under ex post budget balance, and $U^{\beta}_{B}(v_{t}\vert h^{t-1}) = \tilde{U}^{\beta}_{B}(v_{t}\vert h^{t-1})$, $U_{S}^{\beta}(c_{t}\vert h^{t-1}) = \tilde{U}^{\beta}_{S}(c_{t}\vert h^{t-1})$ $\forall$ $v_{t},c_{t}, h^{t-1}$.
\end{corollary}

Corollary \ref{cor_expost_charac} provides a precise characterization of interim transfers or expected utility vectors for the class of all incentive compatible tight mechanisms that implement the efficient allocation under ex post budget balance. Given interim transfers, the ex post transfers can be constructed in a manner analogous to Equations (\ref{transfers_expost0}) and (\ref{transfers_expost}).

Before we end the implementation section, a comment on posted prices and double auctions is in order. When implementing under ex post budget balance, these two mechanisms from the static world immediately come to mind. Posted price mechanisms are robust to informational assumptions but demanding in terms of the underlying equilibrium notion.\footnote{A posted price mechanism is the one where a price is exogenously generated either deterministically or in a random fashion, and the buyer and seller simultaneously decide whether to trade at that price. \citet{hr1987robust} show that in the static bilateral trade problem the only mechanism that satisfies dominant strategy incentive compatibility and ex post individual rationality is the posted price mechanism.} No posted price mechanism can sustain the efficient allocation in a dynamic set up that takes expectations over future utility. Implementing through a double auction seems more hopeful. Though a repetition of a standard static double auction will constraint the mechanism to a form of limited liability, and hence won't sustain efficiency. A mechanism with a fixed payment every period followed by a Markovian double auction seems plausible. This is an interesting question for future work.

\section{Information in dynamic mechanisms}

The equivalence between implementation under interim and ex post budget balance can leave one in a dilemma about the structure of our model. While the former takes expectations over the entire future sequence of transfers, the latter involves no expectation, is independent of future transfers and does not allow for any "leakages" whatsoever. In fact, the total expected utility of the two agents in the model with ex post budget balance in weakly greater than that under interim budget balance, since no part of the economic surplus has to be paid to a third party. A intuitive economic meaning of this equivalence is in making irrelevant the desirability of an intermediary. How and when must the intermediary become salient to the production of the efficient economic surplus in this repeated bilateral trade setting?

The role of an intermediary can become salient if reports to her (or prices in an indirect mechanism) either need to be transparent or can be hidden. I expand on both cases.

\subsection{Robustness to information}
\label{rob_inf}

What if the agents observe reports by the other and before revealing their type or taking a decision to walk away. The incentive and individual rationality constraints must then incorporate a robustness to the current type of the other agent. \citet{athey2007efficiency} motivate it in the following manner: ".. if a player were able to delay his announcement or spy on the other player's information, he might be tempted to make an untruthful announcement once learning about the other player's type." When issues of information leakage or transparency are a concern it may be better to focus on an equilibrium concept where "truthful announcement is incentive compatible regardless of any information he may learn that is not relevant to his own payoff." Hence follows (with-in period) ex post incentive compatibility.
\begin{definition}
A mechanism $m=\left\langle {\bf p,U}\right\rangle$ satisfies ex post incentive compatibility if
\[U_{B}(v_{t},c_{t} \vert h^{t-1}) \geq U_{B}(v'_{t},c_{t};v_{t}\vert h^{t-1})\quad \text{and} \quad U_{S}(v_{t},c_{t}\vert h^{t-1}) \geq U_{S}(v_{t},c'_{t}; c_{t}\vert h^{t-1})\]
$\forall v_{t}, v'_{t} \in \mathcal{V}$, $\forall c_{t}, c'_{t} \in \mathcal{C}$, $\forall h^{t-1} \in H^{t-1}$, $\forall t$.
\end{definition}

Ex post individual rationality can be analogously defined. \citet{athey2007efficiency} and \citet{miller2012collusion} show that efficiency cannot be implemented under ex post budget balance under the stronger notions of ex post incentive compatibility and individual rationality for iid types; it is easy to see that there is therefore no hope for Markovian private information.

In Section \ref{sec_expost} in the appendix, I show that an analogue of Lemma \ref{l_pe}, the history dependent payoff equivalence result, holds for ex post incentive compatibility. Following the construction of the ex post expected utility vectors in Section \ref{sec_charac}, we can state and prove a possibility of efficiency result for ex post incentive compatibility and individual rationality. I omit the proof for it follows the same steps as the proof of Theorem \ref{t_main}.

\begin{proposition}
There exists an ex post incentive compatible and individually rational mechanism that implements the efficient allocation under interim budget balance if and only if $\Pi^{*}(h^{t-1})\geq  0$ $\forall h^{t-1}$, $\forall t$.
\end{proposition}

Therefore, with the help of an intermediary one can ensure robustness to information in the sense just described, which is impossible under ex post budget balance.

\subsection{Hiding information?}

\label{sec_hi}

What if on the other hand the intermediary can hide information from the agents? Would it help from the perspective of efficiency? \citet{myerson1986multi} established that maximal information to the mediator and minimal information to the agents produces the most permissible results. This subsection, considers implementation under interim and ex post budget balance for a mechanism that hides information from the agents.

I will exposit the key ideas through what I call the \emph{intermediate mechanism}. Every period the agents only learn about the other agent's reported type through the realized allocation, that is whether trade happens. Thus, in the event of a trade the buyer only learns that the seller's reported type was less than his, and seller that the buyer's reported type was more than her type. No trade gives the opposite information. Importantly, transfers are observable and must be measurable with respect to this information.

Consider the simple trading problem (STP): $\mathcal{V} = \left\{v_{H},v_{L}\right\}$, $\mathcal{C} = \left\{c_{H},c_{L}\right\}$, and $v_{H}>c_{H}>v_{L}>c_{L}$. Buyer who reports type $v_{L}$ learns the type reported by the seller: it could only have been $c_{L}$ in case of trade or $c_{H}$ in case of no trade, because after any history $h^{t-1}$, $p^{*}(v_{L},c_{L}\vert h^{t-1}) =1$ and $p^{*}(v_{L},c_{H}\vert h^{t-1}) = 0$. On the other hand buyer of type $v_{H}$ learns nothing about the seller's report: $p^{*}(v_{H},c_{H}\vert h^{t-1}) = p^{*}(v_{H},c_{L}\vert h^{t-1}) = 1$. Analogously, if the seller reports type $c_{L}$ she learns the type reported by the buyer, and if she reports $c_{H}$ she does not learn anything.\footnote{For the efficient allocation observing trade is equivalent to observing the actual probability of trade since it takes only two values: 0 or 1. For the second-best these two are not the same for a trade probability between 0 and 1 would mean being able to observe the randomization device of the mediator.}

A little bit of notation follows. For any buyer type $v\in\mathcal{V}$, define
\[\left[v\right]_{1} = \left\{c\in \mathcal{C} \vert c<v\right\} \quad \text{and} \quad \left[v\right]_{0} = \left\{c\in \mathcal{C} \vert c>v\right\} \]
Fix the efficient allocation. Whenever the buyer reports a type $v$ and trade takes place he learns that the seller's type lies in the set $\left[v\right]_{1}\subset\mathcal{C}$, and when trade does not take place he learns that in must lie in $\left[v\right]_{0}\subset\mathcal{C}$. Similarly for any $c\in \mathcal{C}$, define
\[\left[c\right]_{1} = \left\{v\in \mathcal{V} \vert c<v\right\} \quad \text{and} \quad \left[c\right]_{0} = \left\{v\in \mathcal{V} \vert c>v\right\} \]

Starting from $\hat{h}^{1}_{B} = \left\{\hat{v}_{1}\right\}$ and $\hat{h}^{1}_{S} = \left\{\hat{c}_{1}\right\}$, the private histories before the report at time $t+1$ are defined by:
\[\hat{h}^{t+1}_{B} = \left\{h^{t}_{B}, p^{*}_{t}, v_{t}, \left[v_{t}\right]_{p^{*}_{t}}, \hat{v}_{t+1}\right\} \quad \text{and} \quad \hat{h}^{t+1}_{S} = \left\{h^{t}_{S}, p^{*}_{t}, c_{t}, \left[c_{t}\right]_{p^{*}_{t}}, \hat{c}_{t+1}\right\}\]
where $\left(\hat{v}_{t}, \hat{c}_{t}\right)$ and $(v_{t},c_{t})$ are respectively the actual and reported types at time $t$, and $p^{*}_{t} = p^{*}(v_{t},c_{t})$ is the observed efficient trading rule. Note that for the mechanism designer observes the history of all reports: $h^{t} = \left(v^{t},c^{t}\right)$. Moreover, along the truthful history we succinctly define the information partition of the agents as:
\[h^{t+1}_{B} = \left\{h^{t}_{B}, v_{t+1}, p^{*}_{t+1}, \left[v_{t+1}\right]_{p^{*}_{t+1}}\right\} \quad \text{and} \quad h^{t+1}_{S} = \left\{h^{t}_{S}, c_{t+1}, p^{*}_{t+1}, \left[c_{t+1}\right]_{p^{*}_{t+1}}\right\}\]
starting at $h^{1}_{B} = \left\{v_{1}, p^{*}_{1}, \left[v_{1}\right]_{p^{*}_{1}}\right\}$ and $h^{1}_{S} = \left\{c_{1}, p^{*}_{1}, \left[c_{1}\right]_{p^{*}_{1}}\right\}$. The partition of public histories into private history equivalence classes is defined as
\[\left[h^{t}_{B}\right] = \left\{h^{t}=\left(v^{t},c^{t}\right)\in H^{t}\mid c_{\tau}\in \left[v_{\tau}\right]_{p^{*}_{\tau}} \forall \tau\leq t\right\} \text{ and } \left[h^{t}_{S}\right] = \left\{h^{t}=\left(v^{t},c^{t}\right)\in H^{t}\mid v_{\tau}\in \left[c_{\tau}\right]_{p^{*}_{\tau}} \forall \tau\leq t\right\} \]

The min-max dynamic VCG intermediate mechanism defines the expected utility of the agents at public history $h^{t}$ to be
\[U^{**}(h^{t}_{B}) = \mathbb{E}\left[U_{B}^{*}(\tilde{h}^{t})\mid \tilde{h}^{t}\in \left[h^{t}_{B}\right]\right] \quad \text{and} \quad U^{**}(h^{t}_{S}) = \mathbb{E}\left[U_{S}^{*}(\tilde{h}^{t})\mid \tilde{h}^{t}\in \left[h^{t}_{S}\right]\right] \]
where I use the short hand $U_{i}^{*}(h^{t}) = U^{*}_{i}(v_{t},c_{t}\vert h^{t-})$ for $i=B,S$. As a final step, the expected budget surplus generated by the mechanism is given by
\begin{equation*}
\Pi^{**}(h^{t-1}) = \mathbb{E}\left[\sum\limits_{\tau=t}^{T}\delta^{\tau -t}\left(v_{\tau} - c_{\tau}\right)p_{\tau} - U^{**}_{B}(\tilde{h}_{B}^{t}) - U^{**}_{S}(\tilde{h}_{S}^{t}) \mid h^{t-1}\right]
\end{equation*}
Then, we have the characterization of efficiency for the intermediate mechanism.
\begin{proposition} \label{result intermediate mechanism}
There exists an incentive compatible and individually rational intermediate mechanism that implements the efficient allocation under interim budget balance if and only if $\Pi^{**}(h^{t-1})\geq  0$ $\forall h^{t-1}$, $\forall t$.
\end{proposition}

What about ex post budget balance? It is quite straightforward to note that the equivalence between ex post and interim budget balance breaks down for the intermediate mechanism. The measurability restriction on transfers does not allow any more information other than the allocation rule to be revealed to the agents. Under ex post budget balance it culminates into a unique price for trade, and a unique one for no trade, This in turn precludes the possibility of efficiency in most environments of interest. For example, in the (STP), since $v_{H}>c_{H}>v_{L}>c_{L}$, in order to get trade when types realization is $(v_{L},c_{L})$, the price $p$ must be such that $v_{L}\geq p \geq c_{L}$; a restriction to unique price for trade would then immediately rule out the trade for type realizations $(v_{H},c_{H})$. In the dynamic model, trade can be conditioned on how many trades have taken place in the past, however, this static inefficacy remains pervasive.

\subsection{Carefully interpreting a folk proposition}

The above observation is of course limited to the intermediate mechanism. As we allow more permissibility in the measurability restriction, more prices can be employed to sustain trade, however, the most "price discrimination" is achieved in the public mechanism. In fact the public mechanism allows discriminatory prices on the basis of past and current types, resulting in the greatest possibility of efficiency (as long as transfers are observable). While a deeper analysis of dynamic mechanism design with a parametrization of how much information is shared amongst agents is beyond the scope here, we can still point towards a more careful interpretation of the observation made in \citet{myerson1986multi} that \emph{less information sharing is better in terms of allocative efficiency in dynamic models of mechanism design.}  To be sure, there is no error here, Myerson is absolutely correct in pointing out that less information sharing reduces information rents by pooling incentive constraints, however, measurability restrictions imposed by observability of transfers in ex post budget balance is not explicitly considered in that framework.

It is useful to note that if we are only interested in ex ante budget balance, then there is not difference between the intermediate or public mechanism in terms of when efficiency can be sustained. It is straightforward to see that first period expected budget surplus are the same in both mechanisms: $\Pi^{*}=\Pi^{**}$; however the same is not true about $\Pi^{*}(v,c)$ and $\Pi^{**}(v,c)$. In fact there is no obvious way to rank the vector. For example, for the (STP):
\[\Pi^{**}(v_{H},c_{H}) > \Pi^{*}(v_{H},c_{H}), \Pi^{**}(v_{H},c_{L}) < \Pi^{*}(v_{H},c_{L}), \Pi^{**}(v_{L},c_{H}) = \Pi^{*}(v_{L},c_{H}).\]

Using these facts, I make three economically meaningful observations. First, with observability of transfers the most permissive results under ex post budget balance are obtained for the public mechanism. Second, in the event that prices do not communicate all possible underlying information, it is welfare improving to have an intermediary (interim budget balance does better than ex post budget balance). Third, even if the intermediary has the ability to hide information, it may choose not to; its decision will depend on how it evaluates the vectors $\Pi^{*}$ and $\Pi^{**}$.

\section{Conclusion}

I studied the repeated bilateral trade problem and the viability of various institutional rules to sustain efficient trade. In particular I introduced a novel notion of budget balance that allows for an intermediary to broker trade in a self enforcing manner. In the conclusion I briefly explore some questions for future work.

This paper is concerned with implementing the first-best outcome. It would be interesting to look at the second best when the first-best cannot be implemented, for example the allocation rule that maximizes the gains from trade. With general Markovian types, we would run into the problem of binding global incentive constraints (see \citet{battaglini2015optimal}). Solving the AR(1) or two types model should be promising. Relatedly, figuring out the optimal dynamic double auction as a pricing implementation of the second-best allocation rule would be important (see \citet{ad2021} for a related model of dynamic double auction with applications to financial markets).

Understanding how the set implementable allocations change under various institutional arrangements as the level of information flow between agents is made more or less Blackwell informative is an interesting question to pursue in future work. Even when the direction is obvious: less information expands the feasible set, the exact nature or rate of expansion may not be obvious. We observed that the direction of results would also change when ex post budget balance and observability is required.

Interim budget balance does not allow the mechanism designer to draw from past surplus. It can be viewed as a constraint on the intermediary's commitment power. While this would a reasonable assumption in many contexts and an interesting benchmark in it own right, it is also important to note that the family of constraints defining interim budget balance can be generalized to a class where the mechanism designer is allowed to save. This is a potentially important area with connections to the macro-finance literature.

Finally, we assumed here that the agents have deep pockets. It is plausible to think of situations where the agents are cash constrained. There is limited understanding in dynamic mechanism design on the nature of optimal or even feasible allocations under limited liability restrictions (see \citet{chassang2012} and \citet{kl2021} for some recent work). This is also an interesting direction for future work.

\section{Appendix}

\subsection{Tight mechanisms}
\label{sec_tight}

Since we operate in a discrete type space, standard Myersonian techniques are not readily available, see for example \citet{kos2013extreme} on this. In order to get around that problem, for some of the results we will restrict attention to the set of mechanisms where local incentive constraints bind. From \citet*{pavan2014dynamic} and \citet{battaglini2015optimal}, we know that these are indeed the "right" set of mechanisms to consider: under first-order stochastic dominance, monotonic allocations that satisfy local incentive constraints are incentive compatible.

Let $IC^{B}_{i,i-1}(h^{t-1})$ be the incentive compatibility constraint that states that after (public) history $h^{t-1}$, buyer of type $i$ does not want to misreport to be type $i-1$. Similarly, let $IC^{S}_{j,j+1}(h^{t-1})$ be the incentive compatibility constraint that states that after (public) history $h^{t-1}$, seller of type $j$ does not want to misreport to be type $j+1$. Since numerically larger types are better for the buyer and worse for the seller, the correct local incentive constraints to consider in the discrete type space problem are downward ($i\mapsto i-1$) for the buyer and upward ($j\mapsto j+1$) for the seller respectively. Define
\[IC^{local}(h^{t-1}) = \left\{IC^{B}_{i,i-1}(h^{t-1}), IC^{S}_{j,j+1}(h^{t-1}) \mid i=2,...N, j=1,...,M-1 \right\}\]
\[IC^{local} = \bigcup\limits_{t} \left\{IC^{local}(h^{t-1}) \text{ } \forall h^{t-1} \in H^{t-1}\right\}\]
Then, a tight mechanism is simply one where all constraints in $IC^{local}$ hold as equalities.

\begin{definition}
A mechanism $m=\left\langle {\bf p,U}\right\rangle$ is called \emph{tight} if all elements of the set $IC^{local}$ hold as equalities.
\end{definition}

Next, define a monotonic allocation. To operationalize the notion of monotonicity, we put a partial order on the set of histories: for any $h^{t} = \left(v^{t},c^{t}\right)$ and $\widehat{h}^{t} = \left(\hat{v}^{t},\hat{c}^{t}\right)$ $\in H^{t}$, $h^{t}\succeq \widehat{h}^{t}$ if $v_{\tau} \geq \hat{v}_{\tau}$ and $c_{\tau} \leq \hat{c}_{\tau}$ for all $\tau\leq t$.

\begin{definition}
An allocation ${\bf p}$ is monotonic if $h^{t}\succeq \widehat{h}^{t}$ $\Rightarrow$ $p(v_{t}\vert h^{t-1}) \geq p(\hat{v}_{t}\vert h^{t-1})$ and  $p(c_{t}\vert h^{t-1}) \leq p(\hat{c}_{t}\vert h^{t-1})$ for all $v_{t}\geq \hat{v}_{t}$ and $c_{t}\leq \hat{c}^{t}$.
\end{definition}

A well known result then follows. The proof parallels the proof of Proposition 2 in \citet{battaglini2015optimal} and is omitted.

\begin{lemma}
Suppose $m=\left\langle {\bf p,U}\right\rangle$  is such that ${\bf p}$ is monotonic and all constraints in $IC^{local}$ hold as equalities. Then, $m$ is incentive compatible.
\end{lemma}

Since the efficient allocation is monotonic, it follows that all tight mechanisms implementing the efficient allocation are incentive compatible. It is important to note that in the continuous types limit of the model, owing to the dynamic envelope theorem, \emph{all} incentive compatible mechanisms are tight mechanisms.

\subsection{Ex post incentive compatibility and proof of Lemma \ref{l_pe}}

\label{sec_expost}

Recollect the notion of (with-in period) ex post incentive compatibility defined in section \ref{rob_inf}. We use it to state and prove a stronger version of Lemma \ref{l_pe}.

Now, ex post incentive compatibility requires truthtelling for every possible current type of the other agent and every possible public history. This was introduced in \citet{athey2007efficiency}. A partial payoff equivalence result for ex post incentive compatibility follows.

\begin{lemma}
\label{l_pe2}
If $\left\langle {\bf p,U}\right\rangle$ is ex post incentive compatible mechanism, and $\left\langle {\bf p,\tilde{U}}\right\rangle$ is another mechanism such that
\[\tilde{U}_{B}(v_{t},c_{t}\vert h^{t-1}) = U_{B}(v_{t},c_{t}\vert h^{t-1})  + a_{B}(c_{t}, h^{t-1}) \text{ }\forall v_{t},  \text{ and}\]
\[ \tilde{U}_{S}(v_{t},c_{t}\vert h^{t-1})= U_{S}(v_{t},c_{t}\vert h^{t-1})  + a_{S}(v_{t}, h^{t-1}) \text{ }\forall c_{t}\]
for a family of constants $\left(a_{B}(c_{t}, h^{t-1}), a_{S}(v_{t}, h^{t-1})\right)$, then $\left\langle {\bf p,\tilde{U}}\right\rangle$ is also incentive compatible.
\end{lemma}

\begin{proof}
Suppose $\left\langle {\bf p,U}\right\rangle$ is ex post incentive compatible. Fix $h^{t-1}$. Then, $U_{B}(v_{t},c_{t}\vert h^{t-1})$ appears in two kinds of incentive compatibility constraints. First,
\[U_{B}(v_{t},c_{t}\vert h^{t-1}) \geq U_{B}(v'_{t},c_{t}\vert h^{t-1}) + (v_{t}-v'_{t})p(v'_{t},c_{t}\vert h^{t-1})\]
\[ \qquad \qquad \qquad+  \delta \sum\limits_{i=1}^{N} U_{B}(v_{i,t+1}\vert h^{t-1},v'_{t},c_{t}) \left(f(v_{i,t+1}\vert v_{t}) - f(v_{i,t+1}\vert v'_{t})\right)\]
Clearly, for any fixed $c_{t}\in \mathcal{C}$, the addition of the constants $a_{B}(c_{t}, h^{t-1})$ to $U_{B}(v_{t},c_{t}\vert h^{t-1})$ for all $v_{t}\in\mathcal{V}$ does not affect any of these constraints.

Next, fix $v_{t-1}$. The second class of constraints we need to consider are

\[U_{B}(v'_{t-1},c_{t-1}\vert h^{t-2}) \geq U_{B}(v_{t-1},c_{t-1}\vert h^{t-2}) + (v'_{t-1}-v_{t-1})p(v_{t-1},c_{t-1}\vert h^{t-2})\]
\[ \qquad \qquad \qquad+  \delta  \sum\limits_{i=1}^{N} U_{B}(v_{i,t}\vert h^{t-2},v_{t-1},c_{t-1}) \left(f(v_{i,t}\vert v'_{t-1}) - f(v_{i,t}\vert v_{t-1})\right)\]

Again, for any fixed $c_{t}\in \mathcal{C}$, this leads to addition of $a_{B}(c_{t},h^{t-1})$ to $U(v_{t},c_{t}\vert h^{t-1})$ for all $v_{t}\in\mathcal{V}$ which drops out of the constraint.

Therefore, linear additions of constants as defined in the lemma preserves incentives.
\end{proof}

Taking expectations Lemma \ref{l_pe} follows as a corollary. Note that since the original min-max dynamic VCG mechanism constructed in section \ref{sec_charac} is independent of history and satisfies ex post incentive compatibility statically, it is easy to see that it is in fact ex post incentive compatible.

\subsection{Proof of Proposition \ref{p_ante_post} }

If there exists an implementable mechanism that satisfies ex post budget balance, then the same mechanism automatically satisfies interim budget balance. Conversely, suppose $\left\langle {\bf p,U}\right\rangle$ (and equivalently $\left\langle {\bf p,x}\right\rangle$) is an implementable mechanism that satisfies interim budget balance. Thus, $\Pi(h^{t-1})\geq 0$ for all $h^{t-1}$, for all $t$. Define the instantaneous expected budget surplus to be the following:
\[\Pi_{t}(h^{t-1}) = \Pi(h^{t-1}) - \delta \mathbb{E}\left[\Pi(h^{t})\vert h^{t-1}\right] \]
that is, $\Pi_{t}(h^{t-1}) = \mathbb{E}\left[x_{B}(v_{t},c_{t}\vert h^{t-1}) - x_{S}(v_{t},c_{t}\vert h^{t-1}) \mid h^{t-1} \right]$. Define a new mechanism $\left\langle {\bf p,\tilde{U}}\right\rangle$ as follows:
\[\tilde{U}_{B}(v_{t}\vert h^{t-1}) = U_{B}(v_{t}\vert h^{t-1})  + \beta\Pi(h^{t-1}) \text{ }\forall v_{t},  \text{ and}\]
\[ \tilde{U}_{S}(c_{t}\vert h^{t-1})= U_{S}(c_{t}\vert h^{t-1})  + (1-\beta)\Pi(h^{t-1}) \text{ }\forall c_{t}\]
for some number $\beta \in [0.1]$. Then, by Lemma \ref{l_pe} the new mechanism is incentive compatible, and interim budget balance implies it is individually rational.
Further, note that in this new mechanism, $\left\langle {\bf p,\tilde{U}}\right\rangle$, interim budget balance is still satisfied as an equality for every history: $\tilde{\Pi}(h^{t-1}) = 0$ for all $h^{t-1}$, for all $t$. Thus, we have $\tilde{\Pi}_{t}(h^{t-1}) = 0$ for all $h^{t-1}$, for all $t$.

Now, $\tilde{\Pi}_{t}(h^{t-1}) = 0$ gives us a mechanism that satisfies "ex ante budget balance in the static sense" for every history. As a final step, we use the construction typically employed in static mechanism design to create another mechanism $\left\langle {\bf p,\hat{x}}\right\rangle$ (and hence $\left\langle {\bf p,\hat{U}}\right\rangle$) that satisfies ex post budget balance. Define
\[\hat{x}(v_{t},c_{t}\vert h^{t-1}) = \tilde{x}_{S}(c_{t}\vert h^{t-1}) + \left[\tilde{x}_{B}(v_{t}\vert h^{t-1}) - \tilde{x}_{B}(.\vert h^{t-})\right]\]
where
\[\tilde{x}_{B}(.\vert h^{t-1}) = \mathbb{E}\left[x_{B}(v_{t},c_{t}\vert h^{t-1})\mid h^{t-1}\right] \]
Note that
\[\tilde{\Pi}_{t}(h^{t-1}) = \tilde{x}_{B}(.\vert h^{t-1}) - \tilde{x}_{S}(.\vert h^{t-1})= 0\]
and, thus
\[\hat{x}_{B}(v_{t}\vert h^{t-1}) = \tilde{x}_{B}(v_{t}\vert h^{t-1}) \text{  and  } \hat{x}_{S}(c_{t}\vert h^{t-1}) = \tilde{x}_{S}(c_{t}\vert h^{t-1})\]
Therefore, the mechanism $\left\langle {\bf p,\hat{x}}\right\rangle$, with ${\bf \hat{x}_{B} = \hat{x}_{S} = \hat{x}}$ satisfies ex post budget balance, and by construction, since the interim transfers and interim expected utility vectors in the mechanisms are the same as in $\left\langle {\bf p,\tilde{x}}\right\rangle$ (and $\left\langle {\bf p,\tilde{U}}\right\rangle$), it is implementable.

\subsection{Proof of Theorem \ref{t_main}}

The proof proceeds as described in the construction of the min-max dynamic VCG mechanism in section \ref{sec_charac}. Define a subset of constraints which will be used in the relaxed problem: incentive constraints in $IC^{local}$ and individual rationality constraints of the "lowest" types.
\[{\bf C}^{RP} = IC^{local} \bigcup \left\{IR^{B}_{1}(h^{t-1}), IR^{S}_{M}(h^{t-1}) \forall h^{t-1}\in H^{t-1} \forall t\right\} \]
First we show that for the dynamic VCG mechanism all constraints included in $IC^{local}$ are satisfied as equalities.
\begin{lemma}
\label{lem_vcg_eq}
For the dynamic VCG mechanism, $\left\langle {\bf p^{*},U^{vcg}}\right\rangle$, all constraints in $IC^{local}$ hold as equalities.
\end{lemma}
\begin{proof}
We establish the result for the buyer and the result for the seller follows analogously. Fix a history $h^{t-1}$ and buyer type $i>1$. Let $u^{vcg}_{B}(v_{i},c_{j})$ be the static ex post utility of the buyer of type $v_{i}$ when the seller's type is $c_{j}$ in the VCG mechanism (modified for discrete types), that is $u^{vcg}_{B}(v_{i},c_{j}) = v_{i}p^{vcg}(v_{i},c_{j}) - x^{vcg}(v_{i},c_{j})$.\footnote{Note that we have suppressed $h^{t-1}$ in this term because the VCG mechanism is independent of history, and $v_{i,t}$ and $c_{j,t}$ are written as $v_{i}$ and $c_{j}$ for simplicity of notation.} We make the following claim: $u^{vcg}_{B}(v_{i},c_{j}) = \Delta v_{i} p^{vcg}(v_{i-1},c_{j}) + u^{vcg}_{B}(v_{i-1},c_{j})$, where recollect $\Delta v_{i} = v_{i} - v_{i-1}$.

If $v_{i}<c_{j}$, then $p^{vcg}(v_{i-1},c_{j}) = 0$, and $u^{vcg}_{B}(v_{i},c_{j}) = u^{vcg}_{B}(v_{i-1},c_{j}) =0$ and claim holds trivially. Thus, let $v_{i}>c_{j}$. We consider two subcases: (i) $v_{i-1}>c_{j}$, and (ii) $v_{i-1}<c_{j}$.\footnote{Recollect that $\mathcal{V}\cap \mathcal{C} = \emptyset$, so $v_{i-1} = c_{j}$ is ruled out by assumption.} Suppose $v_{i-1}>c_{j}$. Then, it is easy to see that by definition, $p^{vcg}(v_{i},c_{j}) =p^{vcg}(v_{i-1},c_{j}) =1$, and $x^{vcg}_{B}(v_{i},c_{j}) = x^{vcg}(v_{i-1},c_{j})$. Thus,
\[u_{B}^{vcg}(v_{i},c_{j}) = v_{i} - x^{vcg}_{B}(v_{i},c_{j}) = (v_{i} - v_{i-1}) + v_{i-1} - x^{vcg}_{B}(v_{i-1},c_{j}) = \Delta v_{i}p^{vcg}(v_{i-1},c_{j}) + u_{B}^{vcg}(v_{i-1},c_{j})\]
Next, suppose $v_{i}>c_{j}$ but $v_{i-1}<c_{j}$. Then $u_{B}^{vcg}(v_{i-1},c_{j}) =0$, and by construction $x^{vcg}_{B}(v_{i},c_{j}) = v_{i}$. Thus,
\[u_{B}^{vcg}(v_{i},c_{j}) = v_{i} - x^{vcg}_{B}(v_{i},c_{j}) =v_{i} - v_{i} =0 = \Delta v_{i}p^{vcg}(v_{i-1},c_{j}) + u_{B}^{vcg}(v_{i-1},c_{j})\]
Next,
\begin{align*}
& U^{vcg}_{B}(v_{i},c_{j}\vert h^{t-1}) = u^{vcg}(v_{i},c_{j}) + \sum\limits_{k=1}^{N} f(v_{k}\vert v_{i}) U_{B}^{vcg}(v_{k}\vert h^{t-1}, v_{i},c_{j})\\
& = \Delta v_{i}p^{vcg}(v_{i-1},c_{j}) + u_{B}^{vcg}(v_{i-1},c_{j}) + \sum\limits_{k=1}^{N} f(v_{k}\vert v_{i}) U_{B}^{vcg}(v_{k}\vert h^{t-1}, v_{i},c_{j})\\
& = \Delta v_{i}p^{vcg}(v_{i-1},c_{j}) + u_{B}^{vcg}(v_{i-1},c_{j}) + \sum\limits_{k=1}^{N} f(v_{k}\vert v_{i}) U_{B}^{vcg}(v_{k}\vert h^{t-1}, v_{i-1},c_{j})\\
& = \Delta v_{i}p^{vcg}(v_{i-1},c_{j}) + U^{vcg}_{B}(v_{i-1},c_{j}\vert h^{t-1}) + \sum\limits_{k=1}^{N}\left(f(v_{k}\vert v_{i})- f(v_{k}\vert v_{i-1})\right)U_{B}^{vcg}(v_{k}\vert h^{t-1}, v_{i-1},c_{j})
\end{align*}
The second to last equality follows from the fact the mechanism is stationary and as long expectations about future type realizations are the same (which in this case depend only on $c_{t}=c_{j}$) the expected utility vectors are equal. Thus, $U^{vcg}_{B}(v_{k}\vert h^{t-1}, v_{i},c_{j})  = U^{vcg}_{B}(v_{k}\vert h^{t-1}, v_{i-1},c_{j})$.

Above, we showed that the selected constraint holds as an equality for the within period ex post expected utility vector. Finally taking expectation over the seller's cost, we can prove the constraint also holds as an equality for the interim expected utility vector:
\[U^{vcg}_{B}(v_{i}\vert h^{t-1}) = \Delta v_{i}p^{vcg}(v_{i-1}) + U^{vcg}_{B}(v_{i-1}\vert h^{t-1}) + \sum\limits_{k=1}^{N}\left(f(v_{k}\vert v_{i})- f(v_{k}\vert v_{i-1})\right)U_{B}^{vcg}(v_{k}\vert h^{t-1}, v_{i-1})\]
\end{proof}
Next, all constraints in $C^{RP}$ hold as equalities for the min-max dynamic VCG mechanism.
\begin{lemma}
\label{lemK}
For the min-max dynamic VCG mechanism, $\left\langle {\bf p^{*},U^{*}}\right\rangle$, all constraints in $C^{RP}$ hold as equalities.
\end{lemma}
\begin{proof}
The proof for the fact that all constraints in $IC^{local}$ hold as equalities follows from Lemma \ref{lem_vcg_eq} and Lemma \ref{l_pe}. Moreover, by construction $U^{vcg}_{B}(v,c\vert h^{t-1})$ is monotonically increasing in $v$, and $U^{vcg}_{S}(v,c\vert h^{t-1})$ is monotonically decreasing in $c$. Therefore, the respective infima are reached at type 1 and type $M$ respectively, which by construction get zero utility.
\end{proof}
Recall the expression for the expected budget surplus
\[\Pi(h^{t-1}) = \mathbb{E}^{m}\left[\sum\limits_{\tau=t}^{T}\delta^{\tau -t}\left(v_{\tau} - c_{\tau}\right)p_{\tau} - U_{B}(v_{t}\vert h^{t-1}) - U_{S}(c_{t}\vert h^{t-1}) \mid h^{t-1}\right]\]
and define the expected budget surplus from the min-max dynamic VCG mechanism to be $\Pi^{*}(h^{t-1})$.

Fix the efficient allocation ${\bf p^{*}}$. For each history $h^{t-1}$ solve the following relaxed problem $\left(RP(h^{t-1}) \right)$:
\[ \max_{ { \bf U}} \quad \Pi(h^{t-1})\]
subject to
\[\left\langle {\bf p^{*},U}\right\rangle \in {\bf C}^{RP}\]

\begin{lemma}
\label{lem_bind}
In the relaxed problem $(RP(h^{t-1}))$, we have that $IR^{B}_{1}(h^{t-1})$, $IR^{S}_{M}(h^{t-1})$, and all constraints in $IC^{local}(h^{t-1})$ hold as equalities at the optimum.
\end{lemma}
\begin{proof}
First we show that $IR^{B}_{1}(h^{t-1})$ binds at the optimum. If not then, we can decrease $U_{B}(v_{i},c\vert h^{t-1})$ slightly for all $i=1,2,..N$ and all $c\in \mathcal{C}$. This change does not violate any constraints in $(RP(h^{t-1}))$ and moreover strictly increases $\Pi(h^{t-1})$, giving us a contradiction. Thus, we must have that $IR^{B}_{1}(h^{t-1})$ binds at the optimum. Analogously, we can show that $IR^{S}_{M}(h^{t-1})$ binds at the optimum.

Now, for any $i>1$ suppose $IC^{B}_{i,i-1}(h^{t-1})$ does not bind at the optimum. Then, for each $c\in \mathcal{C}$, decrease $U_{B}(v_{k},c\vert h^{t-1})$ by $\varepsilon$ for each $k\geq i$. If $t = 1$, all the constraints are still satisfied and $\Pi(h^{t-1})$ is strictly higher, giving a contradiction. If $t> 1$, this change does not affect any constraint except $IC^{B}_{j+i,j}(h^{t-2})$ where $v_{t-1} = v_{j}$.\footnote{Note that if $v_{t-1} = v_{N}$ we are done, so the rest of the proof assumes $j<N$.} The right hand side of this constraint is reduced by $\left[F(v_{i-1}\vert v_{j}) - F(v_{i-1}\vert v_{j+1})\right]\varepsilon$ which owing to first-order stochastic dominance is non-negative. Now, repeat the same procedure decreasing $U_{B}(v_{k},c\vert h^{t-2})$ by $\left[F(v_{i-1}\vert v_{j}) - F(v_{i-1}\vert v_{j+1})\right]\varepsilon$ for each $k\geq j+1$. We can keep reducing utility vectors backward till the first period, unless $h^{t-1}$ contains $v_{N}$, in which case the backward iteration ends there, to deduce a non-negative change in $\Pi(h^{t-1})$. Thus, the changes do not violate any of the constraints and keep the objective to a value larger than or equal to that before the change.
\end{proof}

Lemma \ref{lem_bind} uniquely pins down interim expected utility vectors: $U_{B}(v_{t}\vert h^{t-1})$ and $U(c_{t}\vert h^{t-1}$ $\forall (v_{t},c_{t}) \in \mathcal{V}\times \mathcal{C}$ that maximize $\Pi(h^{t-1})$ subject to the relaxed problem for all $h^{t-1}$, for all $t$. Call this mechanism $\left\langle {\bf p^{*},\hat{U}}\right\rangle$. It is to be shown that $\left\langle {\bf p^{*},\hat{U}}\right\rangle$ is implementable, that is, it satisfies all incentive compatibility and individual rationality constraints not included in $C^{RP}$.

From Lemma \ref{lemK} one mechanism that delivers these interim expected utility vectors as in ${\bf \hat{U}}$ is the min-max dynamic VCG mechanism, which by construction is implementable. Thus, must have that $\left\langle {\bf p^{*},\hat{U}}\right\rangle$ is implementable and it delivers the expected budget surplus of $\Pi^{*}(h^{t-1})$ for all $h^{t-1}$ and for all $t$.

Putting it all together we have that the min-max dynamic VCG mechanism achieves the highest expected budget surplus for every possible history in the class of implementable mechanisms which proves the result.

Note that following the same methodology we can prove something stronger- the same result for ex post incentive compatibility and ex post individual rationality.

\subsection{Proof of Corollary \ref{c_limit_d}}

The methodology here is analogous to Athey and Segal [2013] and Yoon [2015]. Though unlike their proofs, owing to the dynamic min-max VCG mechanism, the threshold of $\delta$ required for efficiency to hold will be tightest possible in the construction that follows. Let $\Theta = \mathcal{V}\times\mathcal{C}$. First some notation. Define
\[\pi^{vcg} = \left(x^{vcg}_{B}(v,c) - x^{vcg}_{S}(v,c)\right)_{(v,c)\in \Theta}\]
to be the row vector of flow budget for the mechanism designer from the standard static VCG mechanism (modified for discrete types). Further, let $\mu$ be the row vector of joint prior distribution on $\Theta$ composed of independent marginals $f$ and $g$. Also, let $M$ be the Markov matrix governing the evolution of types in $\Theta$, again composed of the independent marginals $f(.\vert .)$ and $g(.\vert .)$. Let $\mu_{B}(v)$ be the conditional distribution of types in the first period when the buyer's type is $v$, and $\mu_{B}(v\vert c)$ be the conditional distribution of types in the second period when the buyer's current type is $v$ and the seller's type in the first period is $c$.  Similarly, define conditionals for the seller. Let $\mu(.\vert v, c)$ be the conditional distribution in the second period when the first period types are $(v,c)$. Finally, define
\[u_{B}^{vcg} = \left(vp^{vcg}(v,c) - x_{B}^{vcg}(v,c)\right)_{(v,c) \in \Theta}\]
\[u_{S}^{vcg} = \left(x_{S}^{vcg}(v,c) - c p(v,c)\right)_{(v,c) \in \Theta}\]
to be the column vector of flow utilities for the buyer and seller in the VCG mechanism.

Given these, it is easy to see that
\[\Pi^{vcg} = \sum\limits_{t=1}^{\infty} \delta^{t-1} \mu M^{t-1} \pi^{vcg}, \quad \Pi^{vcg}(v,c) = \sum\limits_{t=1}^{\infty} \delta^{t-1} \mu(.\vert v,c) M^{t-1} \pi^{vcg}\]
\[U^{vcg}_{B}(\underline{v}) = \inf\limits_{v\in \mathcal{V}}\sum\limits_{t=1}^{\infty} \delta^{t-1} \mu_{B}(v) M^{t-1} u^{vcg}_{B}, \quad
U^{vcg}_{B}(\underline{v}\vert c) = \inf\limits_{v\in \mathcal{V}}\sum\limits_{t=1}^{\infty} \delta^{t-1} \mu_{B}(v\vert c)M^{t-1} u^{vcg}_{B}\]
\[U^{vcg}_{S}(\overline{c}) = \inf\limits_{c\in \mathcal{C}}\sum\limits_{t=1}^{\infty} \delta^{t-1} \mu_{S}(c)M^{t-1} u^{vcg}_{S}, \quad
U^{vcg}_{S}(\overline{c}\vert v) = \inf\limits_{c\in \mathcal{C}}\sum\limits_{t=1}^{\infty} \delta^{t-1} \mu_{S}(c\vert v)M^{t-1} u^{vcg}_{S}\]
Thus, we have
\[\Pi^{*} =  \sum\limits_{t=1}^{\infty} \delta^{t-1} \mu M^{t-1} \pi^{vcg} + \sum\limits_{t=1}^{\infty} \delta^{t-1} \mu_{B}(\underline{v}) M^{t-1} u^{vcg}_{B} + \sum\limits_{t=1}^{\infty} \delta^{t-1} \mu_{S}(\overline{c})M^{t-1} u^{vcg}_{S}\]
\[\Pi^{*}(v,c) = \sum\limits_{t=1}^{\infty} \delta^{t-1} \mu(.\vert v,c) M^{t-1} \pi^{vcg} + \sum\limits_{t=1}^{\infty} \delta^{t-1} \mu_{B}(\underline{v}\vert c)M^{t-1} u^{vcg}_{B} + \sum\limits_{t=1}^{\infty} \delta^{t-1} \mu_{S}(\overline{c}\vert v)M^{t-1} u^{vcg}_{S}\]
where the $\inf$ has been replaced by the element that attains the infimum. Moreover,
\[x^{vcg}_{B}(v,c) = vp(v,c) - u_{B}(v,c), \quad  x^{vcg}_{S}(v,c) = cp(v,c) + u_{S}(v,c)\]
Thus,
\[\pi^{vcg} = y - (u^{vcg}_{B} + u^{vcg}_{S}) \]
where $y = \left((v-c)p^{vcg}(v,c)\right)_{(v,c) \in \Theta}$, and we have
\[\Pi^{*} =  \sum\limits_{t=1}^{\infty} \delta^{t-1} \mu M^{t-1}y -  \sum\limits_{t=1}^{\infty} \delta^{t-1} \mu M^{t-1}(u^{vcg}_{B} + u^{vcg}_{S})+  \sum\limits_{t=1}^{\infty} \delta^{t-1} \mu_{B}(\underline{v}) M^{t-1} u^{vcg}_{B} + \sum\limits_{t=1}^{\infty} \delta^{t-1} \mu_{S}(\overline{c})M^{t-1} u^{vcg}_{S}\]
and,
\[\Pi^{*}(v,c) =  \sum\limits_{t=1}^{\infty} \delta^{t-1} \mu(.\vert v,c) M^{t-1}y -  \sum\limits_{t=1}^{\infty} \delta^{t-1} \mu(.\vert v,c) M^{t-1}(u^{vcg}_{B} + u^{vcg}_{S})\]
\[+  \sum\limits_{t=1}^{\infty} \delta^{t-1} \mu_{B}(\underline{v}\vert c )M^{t-1} u^{vcg}_{B} + \sum\limits_{t=1}^{\infty} \delta^{t-1} \mu_{S}(\overline{c}\vert v)M^{t-1} u^{vcg}_{S}\]

For a representative element $m^{t}_{\theta,\theta'}$ of $M^{t}$ there exists a stationary distribution $\nu$ such that $\nu >> 0$ and $\lim\limits_{t\rightarrow \infty}  m^{t}_{\theta,\theta'}= \nu(\theta')$.\footnote{This is a standard result in Markov chains, see for example Puterman [1994].} Now, fix $\theta$ and $\theta'$ and note that for any $\varepsilon >0$, there exists a $\hat{t}$ such that $\vert m^{t}_{\theta,\theta'}- \nu(\theta') \vert <\varepsilon$ for all $t \geq \hat{t}$.

In what follows, I'm going to establish that $\lim\limits_{\delta\rightarrow 1} \Pi^{*}(v,c)\geq 0$ for all $(v,c) \in\Theta$. We can similarly establish the non-negative limit for $\Pi^{*}$. Fix $(v,c) \in\Theta$, and note that for $t\geq \hat{t}$, we have
\begin{align*}
\mu_{B}(\underline{v}\vert c) M^{t-1} u^{vcg}_{B} &= \sum\limits_{\theta \in \Theta} \sum\limits_{\theta' \in \Theta} \mu_{B}(\underline{v}\vert c)(\theta) m^{t}_{\theta,\theta'}u^{vcg}_{B}(\theta')\\
&> \sum\limits_{\theta \in \Theta} \sum\limits_{\theta' \in \Theta}\mu_{B}(\underline{v}\vert c)(\theta) \left(\nu(\theta') - \varepsilon\right)u^{vcg}_{B}(\theta')\\
&= \sum\limits_{\theta' \in \Theta} \left(\nu(\theta') - \varepsilon\right)u^{vcg}_{B}(\theta')
\end{align*}
where $\mu_{B}(\underline{v}\vert c)(\theta)$ is the probability assigned to $\theta$ by the conditional $\mu_{B}(\underline{v}\vert c)$, etc. Therefore, we have
\[\sum\limits_{t=0}^{\infty} \delta^{t-1} \mu_{B}(\underline{v}\vert c) M^{t-1} u^{vcg}_{B} > \sum\limits_{t=0}^{\hat{t}} \delta^{t-1} \mu_{B}(\underline{v}\vert c) M^{t-1} u^{vcg}_{B} + \frac{\delta^{\hat{t}}}{1-\delta} \sum\limits_{\theta'\in \Theta}\left(\nu(\theta') - \varepsilon\right)u^{vcg}_{B}(\theta')\]
And, similarly for the seller
\[\sum\limits_{t=0}^{\infty} \delta^{t-1} \mu_{S}(\overline{c}\vert v) M^{t-1} u^{vcg}_{S} > \sum\limits_{t=0}^{\hat{t}} \delta^{t-1} \mu_{S}(\overline{c}\vert v) M^{t-1} u^{vcg}_{S} + \frac{\delta^{\hat{t}}}{1-\delta} \sum\limits_{\theta'\in \Theta}\left(\nu(\theta') - \varepsilon\right)u^{vcg}_{S}(\theta')\]
Through the same arguments we have
\[\sum\limits_{t=1}^{\infty} \delta^{t-1} \mu(.\vert v,c) M^{t-1}y > \sum\limits_{t=0}^{\hat{t}}\delta^{t-1}\mu(.\vert v,c) M^{t-1}y + \frac{\delta^{\hat{t}}}{1-\delta} \sum\limits_{\theta'\in \Theta}\left(\nu(\theta') - \varepsilon\right)y(\theta')\]
and,
\[\sum\limits_{t=1}^{\infty} \delta^{t-1} \mu(.\vert v,c) M^{t-1}(u^{vcg}_{B} + u^{vcg}_{S}) < \sum\limits_{t=1}^{\hat{t}} \delta^{t-1} \mu(.\vert v,c) M^{t-1}(u^{vcg}_{B} + u^{vcg}_{S}) + \frac{\delta^{\hat{t}}}{1-\delta} \sum\limits_{\theta'\in \Theta}\left(\nu(\theta') + \varepsilon\right)(u^{vcg}_{B} + u^{vcg}_{S})(\theta')\]
Putting it all together, we have
\[\Pi^{*}(v,c)> \sum\limits_{t=1}^{\hat{t}} \delta^{t-1} \mu(.\vert v,c) M^{t-1}y -  \sum\limits_{t=1}^{\hat{t}} \delta^{t-1} \mu(.\vert v,c) M^{t-1}(u^{vcg}_{B} + u^{vcg}_{S})+  \sum\limits_{t=1}^{\hat{t}} \delta^{t-1} \mu_{B}(\underline{v}\vert c ) M^{t-1} u^{vcg}_{B} \]
\[\qquad + \sum\limits_{t=1}^{\hat{t}} \delta^{t-1} \mu_{S}(\overline{c}\vert v)M^{t-1} u^{vcg}_{S}
+ \frac{\delta^{\hat{t}}}{1-\delta} \sum\limits_{\theta' \in \Theta} \left[\left(\nu(\theta') - \varepsilon\right)y(\theta') -2\varepsilon \left(u_{B} + u_{S}\right)(\theta')\right]\]
Now, owing to finite supports, we can define
\[C_{1} = \max\limits_{\theta\in \Theta} \vert y(\theta) \vert \quad \text{ and } \quad C_{2} = \max\limits_{\theta\in \Theta} \vert \left(u_{B} + u_{S}\right)(\theta) \vert\]
Thus,
\[\sum\limits_{t=1}^{\hat{t}} \delta^{t-1} \mu(.\vert v,c) M^{t-1}y \geq -\frac{1-\delta^{\hat{t}}}{1-\delta}C_{1}\]
\[ \sum\limits_{t=1}^{\hat{t}} \delta^{t-1} \mu(.\vert v,c) M^{t-1}(u^{vcg}_{B} + u^{vcg}_{S}) \leq \frac{1-\delta^{\hat{t}}}{1-\delta}C_{2}\]
and,
\[\sum\limits_{t=1}^{\hat{t}} \delta^{t-1} \mu_{B}(\underline{v}\vert c ) M^{t-1} u^{vcg}_{B} +  \sum\limits_{t=1}^{\hat{t}} \delta^{t-1} \mu_{S}(\overline{c}\vert v)M^{t-1} u^{vcg}_{S} \geq -\frac{1-\delta^{\hat{t}}}{1-\delta}2C_{1}\]
Also, we choose $\varepsilon$ so that
\[\left(\nu(\theta') - \varepsilon\right)y(\theta') -2\varepsilon \left(u_{B} + u_{S}\right)(\theta') > \eta\]
for some $\eta >0$. Therefore
\[\Pi^{*}(v,c) > -\frac{1-\delta^{\hat{t}}}{1-\delta}\left(C_{1} + 3C_{2}\right) + \frac{\delta^{\hat{t}}}{1-\delta}\eta\]
As $\delta \rightarrow 1$, $\frac{1-\delta^{\hat{t}}}{1-\delta}\left(C_{1} + 3C_{2}\right)$ converges to to a finite quantity: $-\hat{t}\left(C_{1}+3c_{2}\right)$, and $\frac{\delta^{\hat{t}}}{1-\delta}\eta$ goes to infinity. Thus, we must have $\lim\limits_{\delta\rightarrow 1} \Pi^{*}(v,c)\geq 0$.

\subsection{Proof of Corollary  \ref{c_limit_a}}

We show that
\[\lim\limits_{\alpha_{B},\alpha_{S}\rightarrow 1} \Pi^{*} <0 \]

Call the min-max dynamic VCG mechanism for $\delta=0$, the min-max static VCG mechanism. It is easy to see that for $\alpha_{B}=\alpha_{S}=1$, that is for constant types, the min-max dynamic VCG mechanism involves the repetition of the min-max static VCG mechanism in every period. For constant types
\[U^{*}_{B}(v,c) = vp^{*}(v,c) - x^{*}_{B}(v,c) + \delta U^{*}_{B}(v,c \vert v,c) \]
and, since ex post utility vectors are independent of history, we have
\[U^{*}_{B}(v,c) = \frac{1}{1-\delta} \left(vp^{*}(v,c) - x^{*}_{B}(v,c)\right)\]
Similarly
\[U^{*}_{S}(v,c) = \frac{1}{1-\delta} \left(x_{S}^{*}(v,c) - cp^{*}(v,c)\right)\]
Thus, it is straightforward to note that
\[\lim\limits_{\alpha_{B},\alpha_{S}\rightarrow 1}\Pi^{*} = \frac{1}{1-\delta}\Pi^{*}(\delta=0) + \lim\limits_{\alpha_{B},\alpha_{S}\rightarrow 1} o\left(\alpha_{B},\alpha_{S}\right) \]
where $o\left(\alpha_{B},\alpha_{S}\right)$ is budget surplus along non-constant histories. Note that transfers along non-constant histories are finite- owing to finite supports and construction of the min-max dynamic VCG mechanism. Moreover, the probabilities of non-constant histories converge to zero. Thus, we must have $\lim\limits_{\alpha_{B},\alpha_{S}\rightarrow 1} o\left(\alpha_{B},\alpha_{S}\right)=0$. Putting this all together, we have
\[\lim\limits_{\alpha_{B},\alpha_{S}\rightarrow 1}\Pi^{*} = \frac{1}{1-\delta}\Pi^{*}(\delta=0) < 0 \]

\subsection{Proof of Corollary  \ref{cor_all_int_mech}}

From Lemma \ref{l_pe}, it is follows that the mechanism $\left\langle {\bf p^{*},U^{\beta}}\right\rangle$ is incentive compatible. It is clearly individually rational. Expected budget surplus after history $h^{t-1}$ is given by
\begin{align*}
\Pi^{\beta}(h^{t-1}) &=  \mathbb{E}\left[\sum\limits_{\tau=t}^{T}\delta^{\tau -t}\left(v_{\tau} - c_{\tau}\right)p^{*}_{\tau} - U^{\beta}_{B}(v_{t}\vert h^{t-1}) - U^{\beta}_{S}(c_{t}\vert h^{t-1}) \mid h^{t-1}\right] \\
&= \Pi^{*}(h^{t-1})- \beta_{B}(h^{t-1})\Pi^{*}(h^{t-1}) - \beta_{S}(h^{t-1})\Pi^{*}(h^{t-1}) \\
&= \left(1- \beta_{B}(h^{t-1})  \beta_{S}(h^{t-1})\right)\Pi^{*}(h^{t-1}) \\
&\geq 0
\end{align*}
Conversely, suppose there exists a tight mechanism $\left\langle {\bf p^{*},U^{\beta}}\right\rangle$ that implements the efficient allocation under interim budget balance. Then, by Lemma \ref{l_pe} it must satisfy
\[U^{\beta}_{B}(v_{t}\vert h^{t-1}) = U^{*}_{B}(v_{t}\vert h^{t-1}) + a_{B}(h^{t-1})\]
\[U^{\beta}_{S}(v_{t}\vert h^{t-1}) = U^{*}_{S}(c_{t}\vert h^{t-1}) + a_{S}(h^{t-1})\]
For individual rationality to be satisfied $a_{B}(h^{t-1})$ and $a_{S}(h^{t-1})$ must be all non-negative constants. Further, if for any history $\Pi^{*}(h^{t-1})=0$, then it must be the case $a_{B}(h^{t-1}) = a_{S}(h^{t-1}) = 0$, else the mechanism will violate interim budget balance. Thus, if $\Pi^{*}(h^{t-1})>0$, there must exist $\beta_{B}(h^{t-1}), \beta_{S}(h^{t-1})\geq 0$ such that
\[a_{B}(h^{t-1}) = \beta_{B}(h^{t-1})\Pi^{*}(h^{t-1}) \quad \text{and} \quad a_{S}(h^{t-1}) = \beta_{S}(h^{t-1})\Pi^{*}(h^{t-1})\]
And, $\Pi^{\beta}(h^{t-1}) \geq 0$ implies $\beta_{B}(h^{t-1}) + \beta_{S}(h^{t-1}) \leq 1$.

\subsection{Proof of Corollary  \ref{cor_expost_charac}}

Suppose $\left\langle {\bf p^{*},U^{\beta}}\right\rangle$ is tight mechanism that implements the efficient allocation under ex post budget balance. Then, by corollary \ref{cor_all_int_mech} we must have that there exist a family of constants $\beta_{B}(h^{t-1})$, $\beta_{S}(h^{t-1})$ in the unit interval such that $\beta_{B}(h^{t-1}) + \beta_{S}(h^{t-1}) \leq 1$ and
\[U^{\beta}_{B}(v_{t}\vert h^{t-1}) = U^{*}_{B}(v_{t}\vert h^{t-1}) + \beta_{B}(h^{t-1})\Pi^{*}(h^{t-1})\]
\[U^{\beta}_{S}(v_{t}\vert h^{t-1}) = U^{*}_{S}(c_{t}\vert h^{t-1}) + \beta_{S}(h^{t-1})\Pi^{*}(h^{t-1})\]
Now, ex post budget balance implies that interim budget balance must be exactly satisfied, that is, $\Pi^{\beta}(h^{t-1}) = 0$ which gives $\beta_{B}(h^{t-1}) + \beta_{S}(h^{t-1}) = 1$.

Conversely, suppose $\left\langle {\bf p^{*},U^{\beta}}\right\rangle$ satisfies the two equations. Then, by Lemma \ref{l_pe} we know that $\left\langle {\bf p^{*},U^{\beta}}\right\rangle$ is incentive compatible, and since we are non-negative constants, it is also individually rational. Now, construct a new mechanism $\left\langle {\bf p^{*},\tilde{x}^{\beta}}\right\rangle$:
\begin{equation*}
\tilde{x}^{\beta}(v,c) = x^{\beta}_{S}(c) + \left[x^{\beta}_{B}(v) - \bar{x}^{\beta}_{B}\right]
\end{equation*}
\begin{equation*}
\tilde{x}^{\beta}(v,c\vert h^{t-1}) = x^{\beta}_{S}(c\vert h^{t-1}) + \left[x^{0}_{B}(v\vert h^{t-1}) - \bar{x}^{0}_{B}(h^{t-1})\right]
\end{equation*}
where $\bar{x}^{\beta}_{B} = \mathbb{E}\left[x_{B}^{\beta}(v,c)\right]$, and $\bar{x}^{\beta}_{B}(h^{t-1}) = \mathbb{E}\left[x_{B}^{\beta}(v,c \vert h^{t-1})\mid h^{t=1} \right]$. Then, by construction, ${\bf \tilde{U}^{\beta}}$ and ${\bf U^{\beta}}$ have the same interim transfers and expected utility vectors, and $\left\langle {\bf p^{*},\tilde{x}^{\beta}}\right\rangle$ satisfies ex post budget balance.

\subsection{Simple trading problem}
\label{STP}

As described in section \ref{sec_cs}, the simple trading problem (\emph{STP}) is referred to the model with two types each for the buyer and seller. Let $\mathcal{V} = \left\{v_{H},v_{L}\right\}$, $\mathcal{C} = \left\{c_{H},c_{L}\right\}$, where $v_{H}>c_{H}>v_{L}>c_{L}$, and denote $f(v_{i}\vert v_{i}) = \alpha_{i}$, and $g(c_{j}\vert c_{j}) = \beta_{j}$.

Define $S_{ij}$ to be the expected surplus associated with with the efficient mechanism when the buyer and seller types are $i$ and $j$ respectively. Then, it is easy to see that
\[S_{HH} = (v_{H} -c_{H}) + \delta\left[\alpha_{H}\beta_{H}S_{HH} + \alpha_{H}(1-\beta_{H})S_{HL} + (1-\alpha_{H})\beta_{H} S_{LH} + (1-\alpha_{H})(1-\beta_{H})S_{LL}\right]\]
\[S_{HL} = (v_{H} -c_{L}) + \delta\left[\alpha_{H}(1-\beta_{L})S_{HH} + \alpha_{H}\beta_{L}S_{HL} + (1-\alpha_{H})(1-\beta_{L}) S_{LH} + (1-\alpha_{H})\beta_{L}S_{LL}\right]\]
\[S_{LH} = 0 + \delta\left[(1-\alpha_{L})\beta_{H}S_{HH} + (1-\alpha_{L})(1-\beta_{H})S_{HL} + \alpha_{L}\beta_{H} S_{LH} + \alpha_{L}(1-\beta_{H})S_{LL}\right]\]
\[S_{LL} = (v_{L} -c_{L}) + \delta\left[(1-\alpha_{L})(1-\beta_{L})S_{HH} + (1-\alpha_{L})\beta_{L}S_{HL}+ \alpha_{L}(1-\beta_{L}) S_{LH} + \alpha_{L}\beta_{L}S_{LL}\right]\]
These are four equations in four unknowns that can be solved in closed form to get values of $S_{HH}, S_{HL},S_{LH}$, and $S_{LL}$.\footnote{The final expressions are involved and calculated using the software \emph{Mathematica}.} Using these the ex ante expected surplus, say $S$, from implementing the efficient allocation is given by:
\[S = \sum\limits_{i=H,L}\sum\limits_{j=H,L} f(v_{i}) g(c_{j}) S_{ij} \]
The within period ex post expected utility of the buyer from the VCG mechanism is defined as follows:
\[U^{B}_{HH} = 0 + \delta\left[\alpha_{H}\beta_{H}U^{B}_{HH} + \alpha_{H}(1-\beta_{H})U^{B}_{HL} + (1-\alpha_{H})\beta_{H} U^{B}_{LH} + (1-\alpha_{H})(1-\beta_{H})U^{B}_{LL}\right]\]
\[U^{B}_{HL} = (v_{H} -c_{L}) + \delta\left[\alpha_{H}(1-\beta_{L})U^{B}_{HH} + \alpha_{H}\beta_{L}U^{B}_{HL} + (1-\alpha_{H})(1-\beta_{L}) U^{B}_{LH} + (1-\alpha_{H})\beta_{L}U^{B}_{LL}\right]\]
\[U^{B}_{LH} = 0 + \delta\left[(1-\alpha_{L})\beta_{H}U^{B}_{HH} + (1-\alpha_{L})(1-\beta_{H})U^{B}_{HL} + \alpha_{L}\beta_{H} U^{B}_{LH} + \alpha_{L}(1-\beta_{H})U^{B}_{LL}\right]\]
\[U^{B}_{LL} = 0 + \delta\left[(1-\alpha_{L})(1-\beta_{L})U^{B}_{HH} + (1-\alpha_{L})\beta_{L}U^{B}_{HL}+ \alpha_{L}(1-\beta_{L}) U^{B}_{LH} + \alpha_{L}\beta_{L}U^{B}_{LL}\right]\]
These are again four equations in four unknowns that can be solved in closed form to get values of $U^{B}_{HH}, U^{B}_{HL},U^{B}_{LH}$, and $U^{B}_{LL}$. Using these we can define the min-max dynamic VCG mechanism as follows:
\[U^{*,B}_{ij} = U^{B}_{ij} - U^{B}_{Lj} \quad \text{ for } i,j =H,L\]
Similarly, we can define $U^{*,S}_{ij}$ for $i,j=H,L$; the within period ex post utility vectors for the seller in the min-max dynamic VCG mechanism. Thus, finally we can write down the expected budget surplus as:
\[\Pi^{*} = S - \sum\limits_{i=H,L}\sum\limits_{j=H,L}f(v_{i})g(c_{j}) \left[U^{*,B}_{ij} + U^{*,S}_{i,j}\right] = \Pi^{vcg} + \sum\limits_{j=H,L}g(c_{j}) U^{*,B}_{Lj} + \sum\limits_{i=H,L}f(v_{i})U^{*,S}_{i,H}\]
and,
\begin{align*}
\Pi^{*}(v_{i},c_{j}) &= \sum\limits\limits_{k=H,L}\sum\limits_{l=H,L}f(v_{k}\vert v_{i})g(c_{l}\vert c_{j})\left[S_{kl} - U^{*,B}_{kl} - U^{*,S}_{k,l}\right]\\
 &= \Pi^{vcg} + \sum\limits_{l=H,L}g(c_{l}\vert c_{j})U^{*,B}_{Ll} + \sum\limits_{i=H,L}f(v_{k}\vert v_{i})U^{*,S}_{k,H}
\end{align*}

For the uniform symmetric simple trading problem the parameter space is simplified in a symmetric fashion: $\mathcal{V} = \left\{1,v\right\}$, $\mathcal{C} = \left\{c,0\right\}$ with $1>c>v>0$ and $\Delta v = 1-v = c = \Delta c$; assume a uniform prior $f=g=\left(\frac{1}{2}.\frac{1}{2}\right)$, and symmetric Markov evolution: $f(v\vert v) = g(c\vert c) = \alpha$ for $\frac{1}{2}<\alpha<1$. It is easy to see that for this parametrization, $U^{*,B} = U^{*,S}$, that is the expected utility vectors are exactly equal. Moreover, $\Pi^{*}(v_{H},c_{H}) = \Pi^{*}(v_{L},c_{L})$.

\subsection{Proof of Proposition  \ref{result intermediate mechanism}}

This follows closely the proof of Theorem \ref{t_main}, so I simply lay out the analogous steps. Start first my stating tight mechanisms analogous to Section \ref{sec_tight}, wherein local incentive constraints bind at the optimum for both the buyer and seller. Then a version of Lemma \ref{l_pe}, that is the payoff equivalence result can be stated for the intermediate mechanism. Finally, the min-max dynamic VCG mechanism described in Section \ref{sec_hi} delivers the precise characterization of the efficient allocation.

\bibliographystyle{abbrvnat}
\bibliography{rb_inter_eff_08February2022}

\end{document}